\documentclass[journal,comsoc]{IEEEtran}
\usepackage{algorithm,algorithmic,amsbsy,amsmath,amssymb,epsfig,bbm,mathrsfs,multirow,amsthm}
\usepackage{bm}
\usepackage{color}
\usepackage{hyperref}
\usepackage{subfigure}

\newtheorem{lemma}{Lemma}
\newtheorem{remark}{Remark}

\newtheorem{proposition}{Proposition}

\begin{document}
%
\title{Optimized Energy and Information Relaying in Self-Sustainable IRS-Empowered WPCN}

\author{Bin Lyu,~\IEEEmembership{Member,~IEEE,}
 Parisa Ramezani,~\IEEEmembership{Student Member,~IEEE,} 
 Dinh Thai Hoang,~\IEEEmembership{Member,~IEEE,}\\
  Shimin Gong,~\IEEEmembership{Member,~IEEE,}
   Zhen Yang,~\IEEEmembership{Senior Member,~IEEE,} 
   and Abbas Jamalipour,~\IEEEmembership{Fellow,~IEEE}
\IEEEcompsocitemizethanks{\IEEEcompsocthanksitem B. Lyu and Z. Yang are with Key Laboratory of Ministry of Education in Broadband Wireless Communication and Sensor Network Technology, Nanjing University of Posts and Telecommunications, Nanjing 210003, China (email: blyu@njupt.edu.cn, yangz@njupt.edu.cn).\protect
\IEEEcompsocthanksitem P. Ramezani and A. Jamalipour are with School of Electrical and Information Engineering, University of Sydney, Sydney, NSW 2006, Australia (email: parisa.ramezani@sydney.edu.au, a.jamalipour@ieee.org).
\IEEEcompsocthanksitem D. T. Hoang is with School of Electrical and Data Engineering, University of Technology Sydney, Sydney, NSW 2007, Australia (email: hoang.dinh@uts.edu.au).
\IEEEcompsocthanksitem S. Gong is with School of Intelligent Systems Engineering, Sun Yat-sen University, China, and also with Peng Cheng Laboratory, Shenzhen 518055, China (e-mail: gongshm5@mail.sysu.edu.cn).
}}

\markboth{IEEE Transactions on Communications,~Vol.~XX, No.~XX, April~2020}%
{Shell \MakeLowercase{\textit{et al.}}: Bare Demo of IEEEtran.cls for IEEE Communications Society Journals}
%



\maketitle

\begin{abstract}
This paper proposes a hybrid-relaying scheme empowered by a self-sustainable intelligent reflecting surface (IRS) in a  wireless powered communication network (WPCN), to simultaneously improve the performance of downlink energy transfer (ET) from a hybrid access point (HAP) to multiple users and uplink information transmission (IT) from users to the HAP. We propose time-switching (TS)  and power-splitting (PS)  schemes for the IRS, where the IRS can harvest energy from the HAP's signals by switching between energy harvesting and signal reflection in the TS scheme or adjusting its reflection amplitude in the PS scheme. For both the TS and PS schemes, we formulate the sum-rate maximization problems by jointly optimizing the IRS's phase shifts for both ET and IT and network resource allocation. To address each problem's non-convexity, we propose a two-step algorithm to obtain the near-optimal solution with high accuracy. To show the structure of resource allocation, we also investigate the optimal solutions for the schemes with random phase shifts. Through numerical results, we show that our proposed schemes can achieve significant system sum-rate gain  compared to the baseline scheme without IRS.
\end{abstract}


%
\IEEEpeerreviewmaketitle

\section{Introduction}
With nearly 50 billion Internet of Things (IoT) devices by 2020 and even 500 billion by 2030 \cite{IoT}, we have already stepped into the new era of IoT. Having the vision of being self-sustainable, IoT has observed the energy limitation as a major issue for its widespread development. Recent advances in energy harvesting (EH) technologies, especially radio frequency (RF) EH \cite{RFSurvey}, opened a new approach for self-sustainable IoT devices to harvest energy from dedicated or ambient RF sources. 
This has led to the emergence of wireless powered communication networks (WPCNs), in which low-cost IoT devices can harvest energy from a dedicated hybrid access point (HAP) and then use the harvested energy to transmit data to the HAP \cite{JuOne}. The development of WPCNs has been a promising step toward the future self-sustainable IoT networks \cite{Abbas}.

Although possessing significant benefits and attractive features for low-cost IoT networks, WPCNs are facing some challenges  which need to be addressed before they can be widely deployed in practice.
In particular, the uplink  information transmission (IT) of IoT devices in WPCNs relies on their harvested energy from downlink energy transfer (ET) of the HAP. However, the IoT devices typically suffer from doubly attenuations of RF signal power over distance \cite{JuOne}, which severely limits the network performance. Reducing the distance between the HAP and IoT devices is one solution to enhance EH efficiency  and achieve greater transmission rates. However, this is not a viable option because IoT devices are randomly deployed in practice, and thus we may not be able to control all of them over their locations. Hence, more efficient and cost-effective solutions are required to enhance the downlink ET efficiency and improve the uplink transmission rate for WPCNs in order to guarantee that WPCNs can be seamlessly fitted into the IoT environment with satisfying performance.

Relay cooperation is an efficient way to enhance the performance of WPCNs, which can be classified into two categories of active relaying and passive relaying. Active relaying refers to scenarios in which the communication between a transmitter and its destined receiver is assisted by a relay which forwards the user’s information to the destination via active RF transmission \cite{HeChen}-\cite{ZengTwo}. However, active relaying schemes have several limitations. Particularly, EH relays need to harvest sufficient energy from the RF sources and use the harvested energy to actively forward information to the receiver. Due to the high power consumption of active relays, it may take a  long time for the relays to harvest enough energy. This thus reduces the IT time of the network. Moreover, most active relays operate in the half-duplex mode, which further shortens the effective IT time, resulting in a network performance degradation. Full-duplex (FD) relays can relax this issue; however, complex self-interference (SI) cancellation techniques are needed at the FD relays to ensure that the SI is effectively mitigated \cite{CJZhong}. In addition, the number of antennas at EH relays is usually limited due to hardware constraints, which also leads to  a  limited performance enhancement. Passive relaying exploits the idea of backscatter communication (BackCom)  for assisting in the source-destination communication \cite{LyuOne}-\cite{Gong}. Specifically, BackCom relay nodes do not need any RF components as they passively backscatter the source’s signals to strengthen the received signals at the receiver. Accordingly, the power consumption of BackCom relay nodes is extremely low and no dedicated time is needed for the relays’ EH \cite{LiuOne}. Nonetheless, as no active signal generation is involved and the passive relays simply reflect the received signal from the source, passive relaying schemes suffer from a poor performance. 

Intelligent reflecting surface (IRS), consisting of a large number of low-cost reflecting elements, has recently emerged as a promising solution to improve the performance of wireless communication networks \cite{Wu2020Survey,Gong2019Towards}.  This technology enables transmitting information without any need for generating new signals but recycling the existing ones \cite{Renzo}. In this way, IRS can adjust the communication environment and create favorable conditions for energy and information transmission without using  energy-hungry RF chains. Having the capability of cooperating in downlink ET and uplink IT, IRS has several advantages over the conventional active and passive relaying techniques \cite{Wu2020Survey}. First of all, IRS is a cost-effective technology and it can be readily integrated into existing wireless communication networks without incurring high implementation costs. Furthermore, IRS is more energy- and spectrum-efficient as compared to conventional relaying methods because it consumes very low power and uses the limited spectrum resources more efficiently. IRS essentially works in the full-duplex mode without causing any interference and adding thermal noise, which further improves the spectral efficiency. Moreover, it is easy to increase the number of IRS  elements to achieve higher performance gains.
\subsection{Background}
IRS has recently been applied to various wireless communication networks and demonstrated promising results for improving the performance in terms of spectrum efficiency (SE) and energy efficiency (EE). References \cite{Huang2018Conf} and \cite{Huang2019IRS} consider the EE maximization problem in an IRS-assisted multi-user downlink communication network. The authors jointly optimize the power allocation at the AP and the phase shifts at the IRS and show that the proposed IRS-assisted communication remarkably outperforms the conventional relay-assisted communication in terms of EE. The authors in \cite{Taha} propose an architecture where a few IRS elements are assumed to be active. Based on the proposed architecture, the authors develop two solutions using compressive sensing and deep learning for designing IRS's reflection matrices. The authors in \cite{HuangDeep} exploit deep reinforcement learning based algorithms to jointly design the transmit beamforming at the base station and phase shifts at the IRS to maximize the sum-rate of downlink   multiple-input  single-output  (MISO) systems.
A low-complexity channel estimation protocol is proposed in \cite{Mishra}, which does not require any prior knowledge on channel state information (CSI) or any active participation from IRS. The authors  design the near-optimal active beamforming at the power beacon and passive beamforming at the IRS in order to maximize the received power at an EH user.  In \cite{Yu}, the authors propose two efficient algorithms for finding the optimal beamformer at the AP and phase shifts at the IRS in an IRS-assisted MISO communication system. The authors in \cite{WuIRS} study the problem of transmit power minimization in a multi-user downlink communication network by jointly optimizing the active transmit beamforming at the AP and passive reflect beamforming at the IRS subject to the users' individual signal-to-interference-plus-noise ratio (SINR) constraints. Compared to the conventional massive MIMO system, the proposed IRS-enhanced model in \cite{WuIRS} can considerably reduce the required transmit power. \cite{WuSWIPT} and \cite{PanSWIPT} study the integration of IRS with simultaneous wireless information and power transfer (SWIPT) technology, where the transmit precoders at the AP and the passive beamforming at the IRS are jointly optimized for maximizing the weighted sum-power at EH users \cite{WuSWIPT} and weighted sum-rate  at information receiving users \cite{PanSWIPT}. Physical layer security and outage probability analysis in IRS-assisted MISO networks are  investigated in \cite{Chu} and \cite{Guo}, respectively. Signal-to-noise-ratio (SNR) maximization problem in a self-sustainable single-user IRS-assisted MISO communication system is studied in \cite{GongTwo}, where IRS elements use part of the downlink information signal for harvesting their required energy. 

A survey on recent research efforts in the area of IRS can be found in \cite{Zhao}.

\subsection{Motivations}
Although IRS has lately received significant interests from the research community, it is still at the very early stage of development and more investigations are needed to fully capture the potentials of IRS and make it applicable to practical scenarios. Specifically, the integration of IRS technology with WPCN is a great step toward the realization of efficient and self-sustainable IoT networks, which has not been well investigated in the literature. Recently, a few research works have investigated the application of IRS for improving the performance of WPCNs \cite{LyuTwo,SuzhiBi}. In \cite{LyuTwo}, the authors study  the application of IRS for WPCN performance enhancement, where IRS elements assist in downlink ET from the HAP to the users and uplink IT from users to the HAP. The authors in \cite{SuzhiBi} propose a similar idea to use the IRS as a hybrid energy and information relay, where the user cooperation is also investigated for a two-user WPCN scenario. These preliminary works on the integration of IRS with WPCN provide some insights on the performance enhancements offered by using IRS in WPCNs. However, this integration needs to be studied more deeply with practical considerations for the network setup and network elements. 

One of the most important points that is often overlooked in the studies on IRS is the IRS's power consumption. Although IRS elements passively reflect the incident signals, the power consumption of the IRS cannot be neglected \cite{Huang2018Conf, Huang2019IRS,GongTwo}. However, the majority of the works in this area (e.g., \cite{Taha}-\cite{Guo}) assume that the IRS's power consumption is negligible because it does not perform complex signal processing tasks. In practice, the power consumption of IRS depends on the type and characteristics of its reflecting elements \cite{Huang2018Conf,Huang2019IRS}. For example, the values of each reflecting element's circuit power consumption are 1.5 and 6 mW for 3- and 5-bit resolution phase shifting, respectively \cite{Huang2019IRS}. As the number of IRS elements is typically large, the circuit power consumption of the IRS can be even comparable to its power supply and cannot be neglected.

In the self-sustainable IoT networks, devices are expected to operate in an uninterrupted manner and have theoretically perpetual lifespans. Considering the non-negligible power consumption of IRS elements, it is important to propose efficient strategies which can keep the IRS operational for very long periods. Although embedded batteries can power the IRS temporarily, they cannot be relied on for the long-term functionality and uninterrupted operation of the IRS. Wired charging  may also be unavailable if the IRS is deployed in inaccessible places.
 Thus, equipping IRS elements with EH modules can resolve these issues and make the IRS energy-neutral \cite{Renzo,HuangHolo}. This is our main motivation for studying a self-sustainable IRS-empowered WPCN, where the EH-enabled IRS, powered by energy transmission of the HAP, can act as a hybrid energy and information relay assisting in both downlink ET and uplink IT. 

\subsection{Contributions}
 We study a self-sustainable IRS-empowered multi-user WPCN, where the IRS is equipped with an EH circuit to harvest RF energy from the HAP to power its operations. Inspired by the conventional wireless-powered active relays \cite{Nasir}, time-switching (TS) and power-splitting (PS) schemes are proposed to enable the IRS to harvest energy from the RF signals transmitted by the HAP. In the TS scheme, the ET phase is split into two sub-slots, where the IRS harvests energy in the first sub-slot and assists in the downlink ET to the users in the second sub-slot. Compared to the conventional TS scheme  \cite{Nasir}, the proposed TS scheme can efficiently improve the amount of harvested energy at the users. In the PS scheme, the IRS harvests energy from the HAP's signal and assists in the downlink ET to the users by adjusting its  amplitude reflection coefficients in the ET phase. Compared to the conventional PS scheme \cite{Nasir}, the proposed PS scheme can enhance both ET and IT efficiency and is more spectrum-efficient. To make our study applicable to practical systems, we consider a piece-wise linear EH model for the IRS and the users  to account for the saturation behavior of practical EH systems \cite{Panos}-\cite{Schober}. We investigate the problem of sum-rate maximization for both TS and PS schemes and optimize the IRS phase shift design and network resource allocation jointly with EH time and amplitude reflection coefficients of the IRS.

The main contributions of this paper are summarized as follows:

\begin{itemize}
\item{We propose a self-sustainable IRS-empowered WPCN, where a wireless-powered IRS acts as a hybrid relay to  improve the performance of WPCN in both downlink ET from the HAP to the users and uplink IT from users to the HAP.} 
\item{To enable energy collection and hybrid relaying functionalities at the IRS, we propose  more efficient TS and PS schemes, which can enhance the ET efficiency from the HAP to the users and assist in the uplink information transmission. We consider a piece-wise linear EH model for the IRS and the users, which is mathematically tractable and is able to capture the saturation effect of practical energy harvesters.}
\item{We study the system sum-rate maximization problem for the TS scheme by jointly optimizing the IRS's phase shift designs in both ET and IT phases, time allocation for the IRS and users' EH, time allocation for each user's IT,  and  the users' power allocation. To deal with the non-convexity of the formulated problem, we propose a two-step algorithm to achieve the near-optimal solution: in the first step the  phase shifts for the IT are obtained in closed-form, while an efficient method by using one-dimensional search,
semidefinite relaxation (SDR) and Gaussian randomization is designed for optimizing the IRS phase shifts in the ET phase,  time allocation and  power allocation in the second step. In particular, we obtain a closed-form solution for the optimal IRS's EH time  and discuss its implications. }
\item{We then investigate the sum-rate maximization problem for the PS scheme and jointly optimize the IRS's phase shift design in both ET and IT phases, time allocation for the  EH and IT phases, power allocation at the users, and the amplitude reflection coefficient in the EH phase, using a similar two-step algorithm as for the TS scheme. In particular, we analyze the condition for activating the IRS in the PS scheme and obtain the optimal amplitude reflection coefficient as a function of the EH time, from which some interesting observations are revealed.}
\item{Finally, we evaluate the performance of our proposed schemes via numerical simulations which show that our proposed schemes can achieve significant system sum-rate  gain  compared to the baseline WPCN protocol.}
\end{itemize}

\subsection{Organization}
This paper is organized as follows. Section \ref{sysmod} describes the system model of the proposed IRS-empowered WPCN for both TS and PS schemes. Sections \ref{TSMax} and \ref{PSMax} investigate the sum-rate maximization problems  for TS and PS schemes, respectively. Section \ref{Simulation} evaluates the performance of the presented algorithms by conducting numerical simulations and Section \ref{Conclusion} concludes the paper.

\section{System Model }
\label{sysmod}
As illustrated in Fig. \ref{SystemModel}, we consider an IRS-assisted WPCN, consisting of an HAP with stable power supply, $N$  energy-constrained users (denoted by $U_i,~ i=1,\ldots,N$), and an energy-constrained IRS. The IRS and users are each equipped with an EH circuit (rectenna) to harvest energy and an energy storage to store the harvested energy. The HAP serves as a central control point for the network, which coordinates the transmissions among all devices and also has the capability and constant energy supply for performing computational tasks. The HAP and users have single antenna each.\footnote{The model can be straightforwardly extended to the scenario that the HAP is with multiple antennas, which will be briefly discussed in Remark \ref{MultipleAntannas}.} The IRS is composed of $K$ passive reflecting elements, which can be configured to direct the incident signals to desired directions. The IRS assists in both downlink ET from the HAP to the users and uplink IT from the users to the HAP. The EH and energy/information relaying at the IRS are controlled by an attached micro-controller.

The downlink channels from  the HAP to $U_i$, from the HAP to the IRS, and from the IRS to $U_i$ are denoted by $h_{h,i}$, $\bm{h}_r \in \mathcal{C}^{K \times 1}$, and $\bm{h}_{u,i}^H \in \mathcal{C}^{1 \times K}$, respectively. The counterpart uplink channels are denoted by $g_{h,i}$, $\bm{g}_r^H \in \mathcal{C}^{1 \times K}$, and $\bm{g}_{u,i} \in \mathcal{C}^{K \times 1}$, respectively. All channels are assumed to be quasi-static flat fading, which remain constant during one block but may change from one block to another \cite{WuSWIPT}. We assume that the channel state information (CSI) of all links is perfectly known.\footnote{The CSI of all links can be precisely obtained by existing channel estimation techniques \cite{Mishra,ShuguangCui}.  In the future work, the effect of channel estimation errors on system performance will be investigated.}

\begin{figure}[t]
\centering
\includegraphics[width=3.2 in, height = 2.2 in] {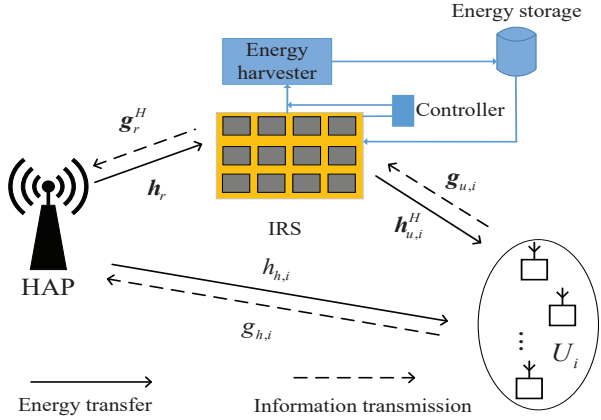}
\caption{System model for an IRS-assisted WPCN.}
\label{SystemModel}
\end{figure} 

\begin{figure}[t]
\centering
\subfigure[ Time-switching scheme.]{
\includegraphics[width=3.4 in, height = 2.2 in]{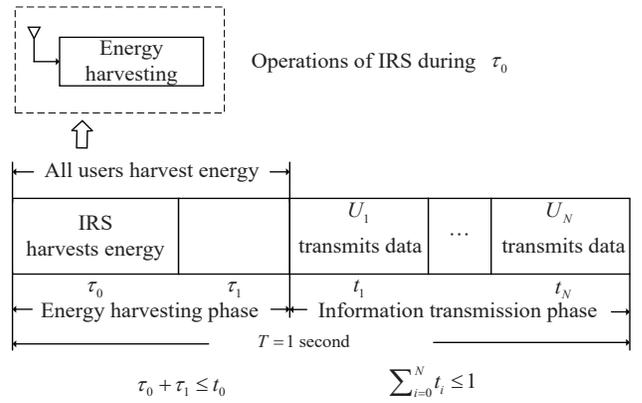}}
\hspace{1in}
\subfigure[Power-splitting scheme.]{
\includegraphics[width=3.4 in, height = 2.2 in]{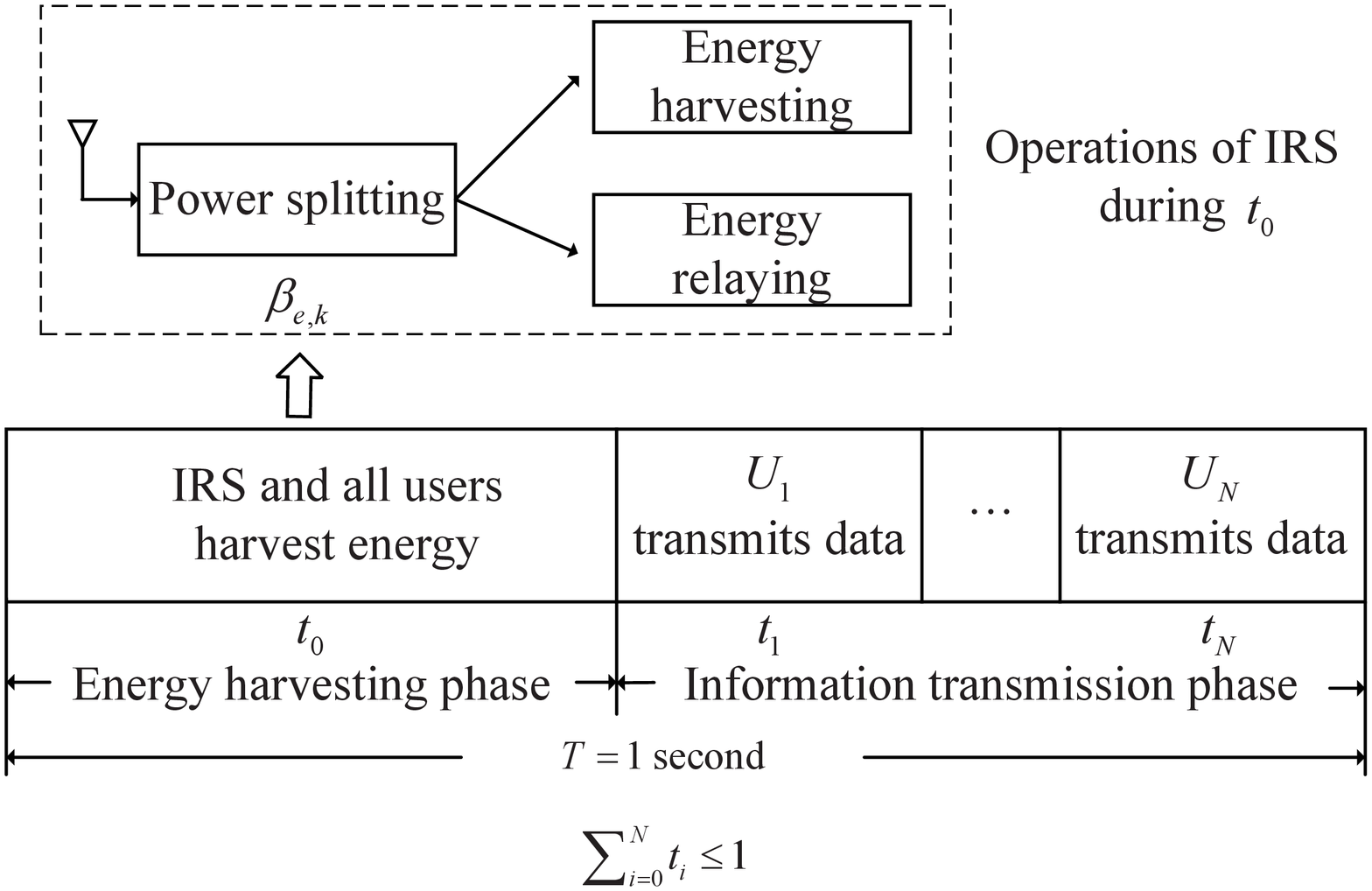}}
\caption{Transmission block structure.}
\label{BlockStructure}
\end{figure}  

The transmission block with a duration of $T$ seconds, is divided into two phases, i.e., ET phase and IT phase. In the ET phase, the HAP transfers energy to the users and IRS in the downlink. The IRS uses the HAP's signals for its own EH and energy relaying to the users. In the IT phase, the users use the harvested energy to transmit data to the HAP with the assistance of the  IRS. Without loss of generality, we consider a normalized unit transmission block time in the sequel, i.e., $T=1$ second.
The details of the ET and IT phases are shown in Fig. \ref{BlockStructure} and elaborated in the following subsections.

\subsection{Energy Transfer Phase}
As mentioned earlier, the IRS is assumed to be energy-constrained, which needs to harvest energy from the HAP for powering its relaying operations. In this regard, we design efficient TS and PS schemes for the IRS.  
\subsubsection{Time-switching scheme}
For the TS scheme, the ET phase with the duration of $t_0$\footnote{The unit of all time coefficients is seconds.} is divided into two sub-slots, having the duration of $\tau_0$ and $\tau_1$, respectively, which satisfy $\tau_0 + \tau_1 \le t_0$. The users can harvest energy over the entire ET phase. For the IRS, it will spend the first sub-slot in the ET phase for its own EH and the second sub-slot for improving the EH efficiency at the users. In particular, in the first sub-slot, all incident signals at the IRS from the HAP are transferred to the EH harvester by setting the amplitude reflection coefficients to be zero, and thus no incident signals will be reflected by the IRS. While
in the second sub-slot, the IRS cooperates with the HAP by adjusting its elements' phase shifts to enhance the total received signal power at the users.
 The transmission block structure for the TS scheme is illustrated in Fig. \ref{BlockStructure} (a).
Denote the transmit signal in the ET phase as 
$x_{h} = \sqrt{P_{h}} s_{h}$,
where $P_{h}$ is the transmit power and $s_{h}$ is the energy-carrying signal with $s_{h} \sim \mathcal{CN}(0,1)$. 

The received signals at the IRS and $U_i$ in the first sub-slot are expressed as
 \begin{align}
&\bm{y}_{r,0} = \bm{h}_r x_{h} + \bm{n}_{r}, \\ 
&y_{ts,0,i} = h_{h,i} x_{h} + n_{u,i},~ i=1,\ldots,N,
 \end{align}
 where $\bm{n}_{r}$ and $n_{u,i}$ denote the additive white Gaussian noises (AWGNs) at the IRS and $U_i$, respectively. Note that the noise power is usually very small and ineffective for EH and can be thus neglected. Hence, the received power at the IRS, denoted by $P_{ts,r}$,  is expressed as $P_{ts,irs} =  P_{h} ||\bm{h}_r ||^2$. Similarly, the  received power at $U_i$ during $\tau_0$ is given by $P_{ts,r,i,0} = P_h |h_{h,i}|^2$.
 
 In the second sub-slot, the IRS assists in the downlink ET.  The phase shift matrix of the IRS during $\tau_1$ is denoted by $\bm{\Theta}_{e} = \sqrt{\rho} \text{diag} \{ \beta_{e,1} e^{j \theta_{e,1}}, \ldots, \beta_{e,K} e^{j \theta_{e,K}} \}$, where $\rho \in (0,1)$  denotes the reflection efficiency and is typically set as a constant \cite{JunZhao}, $\beta_{e,k} \in [0,1]$ and $\theta_{e,k} \in \mathcal{R}$ are the  amplitude reflection coefficient and the phase shift of the $k$-th element, respectively. Let $v_{e,k} = e^{j \theta_{e,k}}$, where $|v_{e,k}|=1$. For the TS scheme, since the IRS only harvests energy during $\tau_0$,  all incident signals at the IRS during $\tau_1$ can be reflected to enhance the EH efficiency, i.e., $\beta_{e,k}=1,~\forall k$ \cite{Wu2020Survey}. Let $\bar{\bm{\Theta}}_e = \text{diag} \{ v_{e,1},\ldots, v_{e,K} \}$. During $\tau_1$, the received signal at $U_i$ for the TS scheme is given by 
 \begin{align}
y_{ts,i} = (\bm{h}_{u,i}^H \sqrt{\rho} \bar{\bm{\Theta}}_{e} \bm{h}_r + h_{h,i}) x_{h} + n_{u,i}, ~i =1,\ldots,N.
 \end{align}
 The received power of $U_i$ during $\tau_1$ is then given by $P_{ts,r,i,1} = P_h |\bm{h}_{u,i}^H \sqrt{\rho} \bar{\bm{\Theta}}_{e} \bm{h}_r + h_{h,i}|^2$.

In practice, the EH circuits usually lead to a non-linear rectification efficiency, i.e., the RF power-to-direct current power conversion is a non-linear function with respect to the received RF power \cite{Boshkovska, Alexandropoulos}. In particular, the harvested power first improves with the increase of received power but finally becomes saturated when the received power is high \cite{Alexandropoulos}.
To approximate the non-linear EH characteristics and account for the saturation region of practical energy harvesters, we employ a two-piece linear EH model,\footnote{There also exist other EH models, e.g., the logistic function based non-linear EH model \cite{Boshkovska} and the multi-piece linear EH model \cite{GuangyueLu}. However, it is noted that the  two-piece linear EH model is sufficiently accurate for modeling the behavior of practical EH circuits. Compared to the logistic function based non-linear EH model, the piece-wise linear EH model is mathematically appealing and easily tractable. In addition, the results obtained from the two-piece linear EH model can be straightforwardly extended to the multi-piece linear EH model.}  which is also widely used in the literature, e.g., \cite{Dong}-\cite{Schober}. According to this model, the harvested power is calculated as 
\begin{align}
P_{h} = \begin{cases}
\eta P_{r}, ~~~\eta P_{r} <P_{sat},  \\ 
P_{sat}, ~~~~~~\text{otherwise},
\end{cases}
\end{align}
where $\eta$ is the EH efficiency in the linear regime,\footnote{In practice, the EH efficiency in this regime is not strictly linear. However, as mentioned in Footnote 4, assuming a constant $\eta$ is still sufficiently accurate for modeling the practical EH circuits.} $P_r$ is the received power, and $P_{sat}$ denotes the saturation power, beyond which there will be no increase in the amount of the harvested power. Therefore, the harvested energy at the IRS and $U_i$ can be obtained as 
\begin{align}
&E_{ts,irs}=\min\{\eta P_{ts,irs}, P_{irs,sat}\} \tau_0 , \\
&E_{ts,u,i}= \min\{\eta P_{ts,r,i,0}, P_{u,i,sat}\} \tau_0 \nonumber \\ 
~~~~&+  \min\{\eta P_{ts,r,i,1}, P_{u,i,sat}\} \tau_1,~i=1,\ldots,N,
\end{align}
where $P_{irs,sat}$ and $P_{u,i,sat}$ represent the saturation power of the IRS and $U_i$, respectively. 

\subsubsection{Power-splitting scheme}
Different from the TS scheme, the dedicated EH time  is not required in the PS scheme and the IRS harvests energy from the HAP by adjusting the amplitude reflection coefficients ($\beta_{e,k},\forall k$)\footnote{Adjusting the reflection coefficient can be achieved by using electronic devices such as positive-intrinsic-negative (PIN) diodes, field-effect transistors (FET), micro-electromechanical system (MEMS) switches, and variable resistor loads \cite{Wu2020Survey,Yang}.}, as illustrated in Fig. \ref{BlockStructure} (b). To be specific, only a part of the HAP's energy signals is fed into the IRS's EH unit for harvesting  and the remaining part is reflected by the IRS to enhance the amount of harvested energy at the users.

It is assumed that all the amplitude reflection coefficients of the IRS elements have the same value, i.e. $\beta_{e,k} = \beta_{e},~ \forall k$.\footnote{In practice, the elements can have different amplitude reflection coefficients. However, the setting will greatly complicate the circuit design of the IRS as different circuits should be integrated to control the amplitude reflection coefficient and phase shift independently at  each element \cite{Wu2020Survey,Yang}. To guarantee the operations of the self-sustainable IRS, we should simplify its circuit design to reduce its circuit power consumption, which can be achieved by setting all amplitude reflection coefficients to be the same.}
  The received signal at $U_i$ in the ET phase for the PS scheme is thus given by
\begin{align}
y_{ps,i} = (\bm{h}_{u,i}^H \sqrt{\rho} \beta_{e}\bar{\bm{\Theta}}_{e} \bm{h}_r  + h_{h,i} ) x_{h} + n_{u,i}, ~i= 1, \ldots,N.
\end{align}
The harvested energy of the IRS and $U_i$ for the PS scheme is then given by 
\begin{align}
&E_{ps,irs} = \min\{\eta P_{h} (1-\beta_{e}^2)  || \bm{h}_r ||^2, P_{irs,sat}\} t_0, \\ 
& E_{ps,u,i} = \min \{\eta P_{h}  | \bm{h}_{u,i}^H \sqrt{\rho} \beta_{e} \bar{\bm{\Theta}}_{e} \bm{h}_r + h_{h,i}  |^2, P_{u,i,sat}\} t_0.
\end{align}

\subsection{Information Transmission Phase}
In the IT phase, the users transmit information to the HAP via time division multiple access, using the harvested energy in the ET phase. Denote the duration of IT for $U_i$ as $t_i$. Let $s_{u,i}$ be the information-carrying signal of $U_i$ with unit power. The transmit signal of $U_i$ during $t_i$ is then expressed as 
$x_{u,i} = \sqrt{P_{u,i}} s_{u,i}$,
where $P_{u,i}$ is $U_i$'s transmit power and satisfies 
\begin{align}
P_{u,i} t_i + P_{c,i} t_i \le E_{f,u,i},~ f = \{ts,ps\},
\end{align} 
with $P_{c,i}$ being the circuit power consumption of $U_i$. As the amplitude reflection coefficients are set to be the same, the IRS's circuit power consumption is mainly caused by performing each element's phase shifting \cite{Huang2018Conf,Huang2019IRS}. The other power consumptions, such as powering the EH circuit and signaling overhead,  can be considered to be negligible \cite{GongTwo,Nasir,Derrick}.  By denoting the power consumption of each element as $\mu$, the circuit power consumption of the IRS is thus expressed as $K \mu$.  To power its operations, IRS needs to harvest sufficient energy in the ET phase.  We assume that all the harvested energy stored in the energy storage can be used to power the IRS' circuits, the following constraints are thus held:
\begin{align}
&K \mu ( \tau_1 + \sum_{i=1}^N t_i) \le E_{ts,irs}, \\ 
& K \mu (t_0 + \sum_{i=1}^N t_i) \le E_{ps,irs},
\end{align}for TS and PS schemes, respectively. Note that the power consumption of the IRS in the first sub-slot of the TS scheme is neglected because the IRS's power consumption is mainly determined by the reflection operation \cite{Huang2018Conf,Huang2019IRS}, which do not take place during $\tau_0$.

Denote  the  phase shift of the $k$-th element  for $U_i$'s IT  as $\theta_{d,i,k} \in \mathcal{R}$. Then,  the phase shift matrix during $t_i$ is denoted by $\bm{\Theta}_{d,i}$, where $\bm{\Theta}_{d,i} = \sqrt{\rho} \text{diag} \{ v_{d,i,1}, \ldots, v_{d,i,K} \}$, $v_{d,i,k} = e^{j \theta_{d,i,k}}$,   and $|v_{d,i.k}|=1$. Note that we have set the  amplitude reflection coefficients to be 1 to maximize the signal reflection in the IT phase \cite{Wu2020Survey}. The received signal at the HAP from $U_i$, denoted by $y_{h,i}$, is thus given by
\begin{align}
y_{h,i} = (\bm{g}_r^H \bm{\Theta}_{d,i} \bm{g}_{u,i} + g_{h,i} ) \sqrt{P_{u,i} } s_{u,i} + n_h,
\end{align}
where $n_h \sim \mathcal{CN} (0, \sigma_h^2) $ is the AWGN at the HAP.
The SNR at the HAP during $t_i $, denoted by $\gamma_{i}$, is expressed as $\gamma_i = \frac{P_{u,i} |\bm{g}_r^H \bm{\Theta}_{d,i} \bm{g}_{u,i} + g_{h,i}|^2 } {\sigma_h^2}$.
The achievable rate from $U_i$ to the HAP in bits/second/Hz is then formulated as
\begin{align}
\label{AchievableRate}
R_i = t_i \log_2 \left(1 + \frac{P_{u,i} |\bm{g}_r^H \bm{\Theta}_{d,i} \bm{g}_{u,i} + g_{h,i}|^2 } {\sigma_h^2} \right).
\end{align}

\section{Sum-rate maximization for the TS scheme}
\label{TSMax}
In this section, we aim to maximize the system sum-rate by jointly optimizing the phase shift design at the IRS in both ET and IT phases, time scheduling of the network, and power allocation at the users. The constraints for the TS scheme are given as follows: $\text{C1:}~ K \mu ( \tau_1 + \sum_{i=1}^N t_i) \le E_{ts,irs}$, $\text{C2:}~ P_{u,i} t_i + P_{c,i} t_i \le E_{ts,u,i},~\forall i$, $\text{C3:}~ \tau_0 + \tau_1 \le t_0$, $\text{C4:}~ \sum_{i=0}^N t_i \le 1$, $\text{C5:}~ \tau_0, \tau_1 \ge 0$, $\text{C6:}~ t_i \ge 0, ~\forall i$, $\text{C7:}~ P_{u,i} \ge 0,~\forall i$,  $\text{C8:}~ |v_{e,k}| =1, ~\forall k$, $\text{C9:} ~ |v_{d,i,k}| = 1, ~\forall i,~\forall k$. The optimization problem is formulated as 
\begin{equation}\tag{$\textbf{P1}$} 
\begin{aligned}
\max_{\bm{\Theta}_{e}, \{\bm{\Theta}_{d,i}\}_{i=1}^N, \bm{t}, \bm{\tau}, \bm{P}_u } ~  &\sum_{i=1}^N R_{i}, \\ 
\text{s.t.}~~~~~~~~
 &\text{C1}-\text{C9},
\end{aligned}
\end{equation}
where $\bm{t} = [t_0,t_1,\ldots,t_N]$, $\bm{\tau} = [\tau_0,\tau_1]$, and $\bm{P}_u = [P_{u,1}, \ldots, P_{u,N}]$.

\subsection{Near-optimal solution to \textbf{P1}}
It is obvious that \textbf{P1} is a non-convex optimization problem due to the coupling of variables in the objective function and the constraints, and convex optimization techniques cannot be used to solve it directly. In the following, we propose a two-step algorithm to solve the sum-rate maximization problem in \textbf{P1}. Specifically, we first obtain the optimal phase shifts for the IT in closed-form and then propose an efficient algorithm to solve the simplified problem.
 
\subsubsection{Optimal phase shift design for IT}
We first present a proposition for the optimal design of phase shifts of the IRS for the IT.
\begin{proposition}
\label{LemmaOne}
The optimal IRS phase shifts for the IT during $t_i$ ($i=1,\ldots,N$) are given by
\begin{align}
\theta_{d,i,k}^* = \arg(g_{h,i}) - \arg(\bm{g}_{r,k}^H) - \arg(\bm{g}_{u,i,k}),~k=1,\ldots,K,
\end{align}
where $\bm{g}_{r,k}^H $ is the $k$-th element of $\bm{g}_{r}^H$, $\bm{g}_{u,i,k}$ is the $k$-th element of $\bm{g}_{u,i}$, and  $\arg(x)$ represents the phase of $x$.
\end{proposition}

\begin{proof}
Refer to Appendix \ref{App:LemmaOne}.
 \end{proof}

 \begin{remark}
\label{RemarkITEnhancement}
From Proposition \ref{LemmaOne}, we can find that there always exists a positive scalar $\delta$ satisfying $|\bm{g}_r^H \bm{\Theta}_{d,i}^* \bm{g}_{u,i}| = \delta |g_{h,i}|$, where $\bm{\Theta}_{d,i}^*$ is  obtained in Proposition \ref{LemmaOne}.
Hence, the received SNR at the HAP during $t_i$ with the assistance of the IRS can be enhanced up to $(1+ \delta)^2$ compared with that without IRS. 
 \end{remark}

\subsubsection{Optimizing phase shift design for ET, time scheduling, and power allocation}
\label{IIIA2}
According to Proposition \ref{LemmaOne}, \textbf{P1} can be simplified as
\begin{equation}\tag{$\textbf{P2}$} 
\begin{aligned}
 \max_{\bm{\Theta}_{e}, \bm{t}, \bm{\tau},   \bm{P}_u} ~ &\sum_{i=1}^N t_i \log_2(1 +  \frac{P_{u,i} \bar{\gamma}_i} {\sigma_h^2}), \\ 
\text{s.t.} ~~& \text{C1}-\text{C8},
\end{aligned}
\end{equation}
where $\bar{\gamma}_i = |\bm{g}_r^H \bm{\Theta}_{d,i}^* \bm{g}_{u,i} + g_{h,i}|^2$. Note that solving \textbf{P2} is equivalent to solving \textbf{P1}.
\textbf{P2} is still non-convex because the variables are coupled in the objective function and the constraints.
To make \textbf{P2} tractable, we introduce $\bm{e}_u=[e_{u,1},\ldots,e_{u,N}]$, where $e_{u,i} = P_{u,i} t_i,~ \forall i$ and set $\bm{\psi}_i = \sqrt{\rho}\text{diag}(\bm{h}_{u,i}^H ) \bm{h}_r$.
Let $\bm{v}_{e} = [v_{e,1}, \ldots, v_{e,K}]^H$, $\bar{\bm v}_{e} = [\bm{v}_{e}^H,1]^H$ and $\bm{V}_{e} = \bar{\bm{v}}_{e} \bar{\bm{v}}_{e}^H$, where $\bm{V}_{e} \succeq 0 $ and $\text{rank} (\bm{V}_{e}) = 1$. 
Based on these new variables, the constraint C2 is recast as  follows:
\begin{align}
& \text{C10:}~e_{u,i} + P_{c,i} t_i \le \min\{\eta P_h |h_{h,i}|^2, P_{u,i,sat} \}\tau_0  \nonumber   \\
&+ \min \Big \{\eta P_h \Big[\text{Tr}(\bm{R}_{e,i} \bm{V}_{e})  + |h_{h,i}|^2 \Big],P_{u,i,sat} \Big\} \tau_1,~\forall i,
\end{align}
where 
$${\bm{R}}_{e,i} = 
\begin{bmatrix}
 \bm{\psi}_{i} \bm{\psi}_i^H & \bm{\psi}_i h_{h,i}^H \\ 
 \bm{\psi}_i^H h_{h,i} & 0 
\end{bmatrix}.
$$
Then, \textbf{P2} can be equivalently rewritten as
\begin{equation}\tag{$\textbf{P2.1}$} 
\begin{aligned}
\max_{\bm{t}, \bm{\tau}, \bm{V}_{e}, \bm{e}_u} ~  &\sum_{i=1}^N t_i \log_2(1 +  \frac{e_{u,i} \bar{\gamma}_i} {t_i \sigma_h^2}), \\ 
\text{s.t.} ~~& \text{C1},~ \text{C3}-\text{C6},~\text{C10}, \\
&\text{C11:}~  e_{u,i} \ge 0,~ \forall i, \\ 
& \text{C12:}~\bm{V}_{e} \succeq 0,  \\ 
& \text{C13:}~ \bm{V}_{e,k,k} = 1, \forall k, \\
&\text{C14:}~\text{rank} (\bm{V}_{e}) = 1.
\end{aligned}
\end{equation}
Due to the rank-one constraint in C14 and coupling of ${\bm{V}}_{e}$ and $\tau_1$ in C10, \textbf{P2.1} is still non-convex and difficult to be solved directly. However, it is straightforward to obtain the optimal duration of the first sub-slot in the ET phase, i.e., $\tau_0$, as  stated in the following proposition.

\begin{proposition}
\label{ProOptiTime}
The optimal duration of the first sub-slot in the ET phase can be obtained as 
\begin{align}
\label{Optitau0}
\tau_0^* = \frac{K \mu}{K \mu + \min \{\eta P_h ||\bm{h}_r||^2,P_{irs,sat}   \} }.
\end{align}
\end{proposition}

\begin{proof}
Refer to Appendix \ref{App:ProOptiTime}.
\end{proof}

\begin{remark}
From Proposition \ref{ProOptiTime},  we can observe that the duration of the first sub-slot in the ET phase is mainly determined by the IRS's setting, e.g., the number of reflecting elements, each element's circuit power consumption, and the saturation power for the EH. With a higher circuit power consumption for each element, the IRS needs more time to harvest sufficient energy, which causes a shorter time for other network operations, e.g., users' EH with the assistance of IRS and users' IT. Furthermore, if the saturation power of IRS ($P_{irs,sat}$) is small such that $\eta P_h ||\bm{h}_r ||^2 \ge P_{irs, sat}$, adding more reflecting elements will increase the IRS's circuit power consumption, which subsequently increases the EH time. Otherwise if $\eta P_h ||\bm{h}_r||^2 < P_{irs, sat}$, increasing the number of elements provides additional transmission links between the HAP and IRS, and thus more energy from the HAP can be transferred to the IRS. Therefore, if the increase of the IRS's circuit power consumption is smaller than that of its harvested power, the EH time of IRS can even be reduced by increasing the number of elements.
\end{remark}

We now proceed to solve \textbf{P2.1} with $\tau_0^*$ obtained in Proposition \ref{ProOptiTime}. For solving \textbf{P2.1}, we first fix $\tau_1$ and optimize time and energy allocation in the IT phase as well as the IRS phase shift design for the ET phase. We can then find the optimal $\tau_1$ by a one-dimensional search over $[0,1-\tau_0)=\big[0, 1- \frac{K \mu} {  K \mu + \min\{\eta P_h ||\bm{h}_r||^2,P_{irs,sat}\} }  \big)$. Denote $\bar{\bm {t}}=[t_1,...,t_N]$.
With fixed $\tau_1$, \textbf{P2.1} is reformulated as 
\begin{equation}\tag{$\textbf{P2.2}$} 
\begin{aligned}
\max_{\bar{\bm{t}}, \bm{V}_{e}, \bm{e}_u} ~  &\sum_{i=1}^N t_i \log_2(1 +  \frac{e_{u,i} \bar{\gamma}_i} {t_i \sigma_h^2}), \\ 
\text{s.t.} ~~&  \text{C6},~\text{C10}-\text{C14}, \\ 
& \sum_{i=1}^N t_i \le 1- \tau_0^* - \tau_1. 
\end{aligned}
\end{equation}

\textbf{P2.2} is still non-convex  due to the rank-one constraint in C14, and its globally optimal solution is thus difficult to obtain. However, using the semidefinite relaxation (SDR) technique \cite{Luo}, we can relax the rank-one constraint to obtain a convex semidefinite programming (SDP) problem \cite{BoydOne}, which can be optimally solved using convex optimization tools, e.g., CVX \cite{BoydTwo}. However, the solution obtained for the relaxed version of \textbf{P2.2} by CVX may not satisfy the rank-one constraint. The Gaussian randomization method is then employed to construct a rank-one solution to \textbf{P2.2} from the solution obtained by CVX. Note that the constructed rank-one solution can be a near-optimal solution to \textbf{P2.2} as it is constructed and searched with quite large times of randomization \cite{SPR}.

Denote the optimal solution to the relaxed problem as  $\{ t_1^*,\ldots,t_N^*, \bar{e}_{u,1},\ldots,\bar{e}_{u,N}, \bar{\bm{V}}_e \}$. The singular value decomposition (SVD) of $\bar{\bm{V}}_{e}$ is expressed as $\bar{\bm{V}}_{e} = \bm{U}_{e} \bm{\varSigma}_{e} \bm{U}_{e}^H$, where $\bm{U}_{e} \in \mathcal{C}^{(K+1) \times (K+1)}$  and $\bm{\varSigma}_{e} \in \mathcal{C}^{(K+1) \times (K+1)}$ are the unitary matrix and diagonal matrix, respectively. Then, the approximate solution to \textbf{P2.2}, denoted by $\hat{\bm{v}}_e$, can be constructed as follows
\begin{align}
\label{RandomVectorTwo}
\hat{\bm{v}}_e = \bm{U}_e \sqrt{\bm{\Sigma_e}} \bm{r}_e,
\end{align}
where $\bm{r}_{e}$ is a random vector with  $\bm{r}_{e} \sim \mathcal{CN} (\bm{0}, \bm{I}_{K+1} )$. Note that as the objective function is an increasing function of $e_{u,i}$, C10 must be an equality at the optimal solution. Therefore, based on the generated random vectors, the energy allocation of the users is computed as
\begin{align}
\label{RecomputedEnergy}
\hat{e}_{u,i} &= \Big( \min\Big\{\eta P_h |h_{h,i}|^2, P_{u,i,sat}\Big\} \tau_0^* +  \nonumber \\ 
&\min\Big\{\eta P_h \Big[\text{Tr}(\bm{R}_{e,i} \hat{\bm{v}}_e \hat{\bm{v}}_e^H) + |h_{h,i}|^2 \Big],P_{u,i,sat}\Big\} \tau_1  \nonumber \\
& - P_{c,i} t_i^*\Big)^+, ~\forall i,
\end{align}
where $(x)^+$ means $\max (x,0) $. We generate $D$ times of random vectors and compute the corresponding objective function values for \textbf{P2.2}. The near-optimal solution to \textbf{P2.2}, denoted by $\hat{\bm{v} }_e^*$, is the one achieving the maximum objective function value. The near-optimal $\bm{v}_e^*$, is finally recovered by
\begin{align}
\bm{v}_e^* = e^{j \arg \Big ( \Big [\frac{\hat{\bm{v}}_e^*}{\hat{\bm{v}}_{e,K+1}^* } \Big ]_{(1:K)} \Big ) },
\end{align}
where $[\bm{\omega}]_{(1:M)}$ represents that the first $M$ elements of $\bm{\omega}$ are taken, $\hat{\bm{v}}_{e,K+1}^*$ denotes the $(K+1)$-th element of $\hat{\bm{v}}_e^*$. It has been numerically and mathematically proved in the literature that the SDR technique followed by Gaussian randomization can provide a good approximation of the optimal solution (see \cite{Luo} and the references therein).

The procedure for solving the sum-rate maximization problem for the TS scheme is summarized in Algorithm  \ref{Alg:Two}, in which the two steps are implemented sequentially. According to \cite{Luo}, the worst-case computational complexity of  Algorithm \ref{Alg:Two} is $\mathcal{O}(M \max(K,N)^4 K^{0.5} \log(1/\epsilon) + MDN )$, where $\epsilon$ is the computational accuracy for using the interior-point method in CVX, and $M$ is the number of iterations for updating $\tau_1$. 
By running Algorithm \ref{Alg:Two} at the HAP, we can obtain the globally-optimal solution for \textbf{P1} approximately,  the analysis of which is given as follows. First, we can obtain the optimal phase shifts for the IT and the optimal duration of the first sub-slot in the ET phase, which are the globally optimal solutions to \textbf{P1}. Second, the globally optimal duration of the second sub-slot in the ET phase can also be found by setting an  appropriate step size. Third, the SDR technique followed by quite large times of randomization based on the Gaussian randomization scheme can guarantee at least $\frac{\pi}{4}$ approximation of the maximum objective function value achieved by solving \textbf{P2.2} \cite{SPR}.

\begin{algorithm}
\caption{ The Algorithm for Solving \textbf{P1}.}
\label{Alg:Two}
\begin{algorithmic}[1] 
\STATE{Initialize $D$ and the step size $\Delta$. Let $\tau_1 = 0$.  }
\STATE{Find the optimal phase shifts for the IT from Proposition \ref{LemmaOne} and optimal $\tau_0$ from Proposition \ref{ProOptiTime}.}
\WHILE{$\tau_1 < 1- \frac{K \mu} {  K \mu + \min\{\eta P_h ||\bm{h}_r||^2,P_{irs,sat}\} }$}
\STATE Solve the relaxed version of \textbf{P2.2} with fixed $\tau_1$ and obtain its optimal solution $\bar{\bm{V}}_{e}$.
\STATE Compute the SVD of $\bar{\bm{V}}_{e}$ and obtain $\bm{U}_{e} $ and $\bm{\varSigma}_{e}$.
 \FOR{$\bar{D}=1:D$}
    \STATE{Generate $\hat{\bm{v}}_{e}$ by \eqref{RandomVectorTwo} and find $\hat{e}_{u,i},~\forall i$ from \eqref{RecomputedEnergy}. }
         \STATE Calculate the objective function value of \textbf{P2.2} and denote it by $R_{\text{sum}}(\bar{D})$.
 \ENDFOR
 \STATE{Set $R^*(\tau_1)=\max R_{\text{sum}}$}.
  \STATE $\tau_1 = \tau_1 +\Delta$.
\ENDWHILE
  \STATE {Set $\tau_1^*=\arg \max_{\tau_1 } R^* $, $\bm{v}_e^*=\bm{v}_e^*(\tau_1^*)$, $\bm{e}_{u}^*=\bm{e}_{u}^* (\tau_1^*)$, $\bm{\bar{t}}^*=\bm{\bar{t}}^* (\tau_1^*)$. }
  \STATE {Set $P_{u,i}^*=e_{u,i}^*/t_i^*,~\forall i$ and extract $\theta_{e,k}^*,~\forall k$ from $\bm{v}_e^*$}.
 \end{algorithmic}
\end{algorithm}

\begin{remark}
\label{MultipleAntannas}
For the multi-antenna HAP scenario, the sum-rate maximization can be achieved by jointly optimizing the phase shift designs, time scheduling, power allocation at the users, and transmit beamforming at the HAP, denoted by $\bm{w}_h$. The formulated problem can be solved by a similar two-step algorithm, where the problem is decoupled into two sub-problems. In particular, the first sub-problem optimizes the phase shifts for the IT, which can be solved by performing the SDR and Gaussian randomization.  The second sub-problem can be solved by using the block coordinate descent (BCD) techniques. Specifically, we can divide the variables into two blocks, i.e., $\bm{w}_h$ and $\{\bm{t}, \bm{\tau}, \bm{\theta}_e, \bm{P}_u \}$, and iteratively optimize one of them with the other one fixed in an alternating manner. A sub-optimal solution can be finally obtained by sequentially solving the two sub-problems. However, the algorithm's complexity for this scenario is much higher, i.e., $\mathcal{O}\Big( \Big(M \max(K,N)^4 K^{0.5} \log(1/\epsilon) + MDN + \min\{N,Q\}^4 Q^{0.5} \log(1/\epsilon) \Big) \log(1/ \varsigma)  \Big)$, where $Q$ is the number of antennas at the HAP, and $\varsigma$ is the designed accuracy for the BCD method. Hence, this scenario causes a much higher implementation cost and is not appropriate for the self-sustainable IoT network.
\end{remark}

\subsection{Random phase shifts with optimized resource allocation for the TS scheme}
\label{RndPhase}
To reduce the computational complexity and show more insights about resource allocation, we consider a special case with random design of phase shifts and focus on the time and power allocation optimization in the IRS-assisted WPCN. As will be shown in Section \ref{Simulation}, using IRS is beneficial for improving the performance of WPCN even with randomly designed phase shifts \cite{WuIRS,Arun}. Letting $e_{u,i} = P_{u,i} t_i $, we have 
\begin{align}
\text{C15:}~ e_{u,i} + P_{c,i} t_i &\le  \min\{\eta P_{ts,r,i,0}, P_{u,i,sat}\} \tau_0 \nonumber \\
&+  \min\{\eta P_{ts,r,i,1}, P_{u,i,sat}\} \tau_1,~\forall i,
\end{align}and the sum-rate maximization problem with random phase shifts is formulated as 
\begin{equation}\tag{$\textbf{P3}$} 
\begin{aligned}
\max_{\bm{t}, \bm{\tau}, \bm{e}_u } ~&  \sum_{i=1}^N t_i \log_2(1+ \frac{\gamma_{d,i}}{\sigma_h^2} \frac{e_{u,i} }{t_i }), \\ 
\text{s.t.} \ \ & \text{C1},~\text{C3}-\text{C6}, ~\text{C11},~\text{C15}, 
\end{aligned}
\end{equation}where $\gamma_{d,i} =|\bm{g}_r^H \bm{\Theta}_{d,i} \bm{g}_{u,i} + g_{h,i}|^2$.
The constraint C15 is an equality at the optimal solution as we discussed in Section \ref{IIIA2}. Hence, we have 
\begin{align}
\label{EqualEnergy}
e_{u,i} &= \min\{\eta P_{ts,r,i,0}, P_{u,i,sat}\} \tau_0 \nonumber \\ 
&+ \min\{\eta P_{ts,r,i,1}, P_{u,i,sat} \} \tau_1 - P_{c,i} t_i ,~\forall i
\end{align}
Substituting \eqref{EqualEnergy} into $R_i$, we have 
\begin{align}{}
R_i = t_i \log_2(1+ \frac{a_i + b_i \tau_1}{t_i} -c_i),
\end{align} 
where $a_i = \min\{\eta P_{ts,r,i,0}, P_{u,i,sat}\} \tau_0 \gamma_{d,i}  / \sigma_h^2$, $b_i = \min\{\eta P_{ts,r,i,1}, P_{u,i,sat}\}  \gamma_{d,i} / \sigma_h^2$ and $c_i = P_{c,i} \gamma_{d,i} / \sigma_h^2$. Proposition \ref{ProOptiTime} holds here as well. Hence, \textbf{P3} is  modified as 
\begin{equation}\tag{$\textbf{P3.1}$} 
\begin{aligned}
\max_{\bar{\bm {t}}, {\tau}_1} ~&  \sum_{i=1}^N t_i \log_2(1+ \frac{a_i + b_i \tau_1}{t_i } -c_i), \\ 
\text{s.t.} \ \ & \text{C6}, ~\tau_1 \ge 0, \\
& \sum_{i=1}^N t_i \le 1- \tau_0^* - \tau_1.
\end{aligned}
\end{equation}

It can be verified that \textbf{P3.1} is a convex optimization problem \cite{BoydOne}, which can be solved by standard convex optimization techniques, e.g., Lagrange duality method. The Lagrangian of \textbf{P3.1} is given by  
\begin{align}
\mathcal{L} (\bar{\bm {t}}, \tau_1, \xi) &= \sum_{i=1}^N t_i \log_2(1+ \frac{a_i + b_i \tau_1}{t_i} -c_i) \nonumber \\
&- \xi \Big[\tau_0^* + \tau_1 + \sum_{i=1}^N t_i -1  \Big],
\end{align}
where  $\xi \ge 0$ is the Lagrange multiplier. 

\begin{proposition}
\label{ProOptiTimeAllocation}
With random design of phase shifts, the optimal time scheduling for the TS scheme is given by 
\begin{align}
\label{Opttau1}
&\tau_1^* = \frac{1-  \frac{K \mu }{K \mu + \min\{\eta P_h ||\bm{h}_r||^2,P_{irs,sat}\} } - \sum_{i=1}^N \frac{a_i}{z_i^* + c_i} }{1+ \sum_{i=1}^N \frac{b_i}{z_i^* + c_i } }, \\
\label{Optti}
& t_i^* = \frac{a_i + b_i\tau_1^*} { z_i^* + c_i }, ~\forall i,
\end{align}
where $z_i^* >0$ is the unique solution of 
$\log_2(1+z_i) - \frac{z_i+c_i}{\ln(2) (1+z_i) } = \xi^*$,
and $\xi^*$ is the optimal dual variable.
\end{proposition}

\begin{proof}
Refer to Appendix \ref{App:ProOptiTimeAllocation}.
\end{proof}

Using \eqref{EqualEnergy} and Proposition  \ref{ProOptiTimeAllocation}, the optimal energy allocation at each user can be easily obtained. 
\section{Sum-rate maximization for the PS scheme}
\label{PSMax}
In this section, we investigate the optimal solution to the sum-rate maximization problem for the PS scheme. The problem is formulated as 
\begin{equation}\tag{$\textbf{P4}$} 
\begin{aligned}
\max_{ \bm{t},  \bm{\Theta}_{e}, \{\bm{\Theta}_{d,i}\}_{i=1}^{N}, \bm{P}_u, {\beta}_e } ~  &\sum_{i=1}^N R_{i}, \\ 
\text{s.t.} \ \ &\text{C4},~\text{C6}-\text{C9},\\ 
& \text{C16:}~ K \mu (t_0 +  \sum_{i=1}^N t_i ) \le E_{ps,irs},~ \forall i, \\
& \text{C17:}~ P_{u,i} t_i + P_{c,i} t_i \le E_{ps,u,i}, ~\forall i, \\
& \text{C18:}~ 0 \le \beta_{e} \le 1.
\end{aligned}
\end{equation}

\subsection{Near-optimal solution to \textbf{P4}}
Similar to \textbf{P1}, \textbf{P4} is a non-convex optimization problem due to the coupled variables in the objective function and the constraints. It is straightforward to observe  that Proposition \ref{LemmaOne} also holds for \textbf{P4}. Accordingly, \textbf{P4} can be  equivalently reformulated as
\begin{equation}\tag{$\textbf{P4.1}$} 
\begin{aligned}
 \max_{ \bm{t},  \bm{\Theta}_{e},  \bm{P}_u, {\beta}_e} ~ &\sum_{i=1}^N t_i \log_2(1 +  \frac{P_{u,i} \bar{\gamma}_i} {\sigma_h^2}), \\ 
\text{s.t.} ~~&\text{C4}, ~\text{C6}-\text{C8}, ~\text{C16}-\text{C18}.
\end{aligned}
\end{equation}

\begin{lemma}
\label{LemmaFeasible}
If the PS scheme is employed at the IRS, the following condition must be met in order to guarantee that IRS can be used for assisting in downlink ET and uplink IT:
\begin{align}
\label{FeasibleCondition}
K \mu < \min\{\eta P_h ||\bm{h}_r ||^2, P_{irs,sat}\}.
\end{align}
\end{lemma}
\begin {proof}
Refer to Appendix \ref{App:LemmaFeasible}.
\end{proof}

\begin{remark}
From Lemma \ref{LemmaFeasible}, we can observe that  the PS scheme cannot always be used, i.e., if \eqref{FeasibleCondition} is not satisfied. That is to say, the applications of the PS scheme are restricted by the IRS's setting (i.e., the number of IRS elements, the circuit power consumption, and the saturation power) and the network environment (i.e., the transmit power at the HAP and  the channel power gain between the HAP and IRS).
If $P_{irs, sat} > \eta P_h \|\bm{h}_r \|^2$,  we can increase the transmit power at the HAP and/or reduce the distance between the HAP and IRS to enable the PS scheme. However, if $P_{irs, sat} \le \eta P_h \|\bm{h}_r \|^2$, the maximum number of IRS elements for enabling the PS scheme is $\lfloor \frac{P_{irs, sat}}{\mu} \rfloor$. Compared to the PS scheme, the TS scheme is free from the limitation and can be applied  more widely.
\end{remark}

In the following, we investigate \textbf{P4.1} under the condition that \eqref{FeasibleCondition} is satisfied, because otherwise the IRS is not able to improve the performance of WPCN. Following the same steps as in Section \ref{IIIA2}, the sum-rate maximization problem is formulated as  

\begin{equation}\tag{$\textbf{P4.2}$} 
\begin{aligned}
 \max_{\bm{t},  \bm{e}_u, \beta_e, \bm{V}_e} ~ &\sum_{i=1}^N t_i \log_2(1 +  \frac{\bar{\gamma}_i} {\sigma_h^2}\frac{e_{u,i}}{t_{i}}), \\ 
\text{s.t.} ~~& \text{C4},~\text{C6},~\text{C11}-\text{C14},~\text{C16},~\text{C18},\\
&\text{C19:}~e_{u,i}+ P_{c,i} t_i \nonumber \\
&~~~\le \min \Big\{\eta P_h [\text{Tr}(\bm{\bar{R}}_{e,i} \bm{V}_{e}) + |h_{h,i}|^2 ],P_{u,i,sat}\Big\} t_0, \\
\end{aligned}
\end{equation} 
where 
$${\bm{\bar{R}}}_{e,i} = 
\begin{bmatrix}
 \beta_e^2\bm{\psi}_{i} \bm{\psi}_i^H & \beta_e\bm{\psi}_i h_{h,i}^H \\ 
 \beta_e\bm{\psi}_i^H h_{h,i} & 0 
\end{bmatrix}.
$$
Similarly, solving \textbf{P4.2} is equivalent to solving \textbf{P4}. From \textbf{P4.2}, we first obtain the following proposition about the optimal amplitude reflection coefficient.

\begin{proposition}
\label{Relationship}
The optimal value of the amplitude reflection coefficient $\beta_e$ is obtained as 
\begin{align}
\label{Opt_Beta}
\beta_e^* = \sqrt{1-\dfrac{K\mu}{\eta P_h||\bm{h}_{r}||^2 t_0^*}},
\end{align}
where $\max \{\frac{K \mu}{\eta P_h ||\bm{h}_r ||^2 }, \frac{K\mu}{P_{irs,sat}}\} < t_0^* < 1$.
\end{proposition}

\begin {proof}
Refer to Appendix \ref{App:LemmaFeasible}.
\end{proof}



For solving \textbf{P4.2}, we first fix $t_0$ and optimize other variables. The optimal value of $t_{0}$ can then be obtained by a one-dimensional search over $\big( \max \{K \mu / (\eta P_h ||\bm{h}_r ||^2), K \mu/P_{irs,sat}\}, 1 \big)$.  
Given $t_0$, the optimal value of $\beta_e$ can be found from Proposition \ref{Relationship} and we will have the following optimization problem:
\begin{equation}\tag{$\textbf{P4.3}$} 
\begin{aligned}
 \max_{\bar{\bm{t}},  \bm{e}_u, \bm{V}_e} ~ &\sum_{i=1}^N t_i \log_2(1 +  \frac{\bar{\gamma}_i} {\sigma_h^2}\frac{e_{u,i}}{t_{i}}), \\ 
\text{s.t.} ~~& \text{C4},~\text{C6},~\text{C11}-\text{C14},~\text{C19}.
\end{aligned}
\end{equation}

After the relaxation of the rank-one constraint in C14,  \textbf{P4.3} is similar to \textbf{P2.2} in Section \ref{IIIA2} and can be solved following the same procedure. For brevity and to avoid repetition, we do not explain the details of solving \textbf{P4.3} here.

Algorithm \ref{Alg:Three} describes the process of solving the sum-rate maximization problem for the PS scheme, which is implemented at the HAP. Similar to what has been mentioned for Algorithm \ref{Alg:Two}, the computational complexity of Algorithm \ref{Alg:Three} is $\mathcal{O}(\bar{M} \max (K,N)^4 K^{0.5} \log(1/\epsilon) + \bar{M}DN )$, where  $\bar{M}$ is the number of iterations for updating $t_0$. Again, by setting the appropriate step size for updating $t_0$ and relatively large number of randomizations for the Gaussian randomization method, we can obtain the near-optimal solution to \textbf{P4}.

\begin{algorithm}
\caption{ The Algorithm for Solving \textbf{P4}.}
\label{Alg:Three}
\begin{algorithmic}[1] 
\STATE{Initialize $t_0 = \max \{\frac{K \mu}{\eta P_h ||\bm{h}_r ||^2 }, \frac{K\mu}{P_{irs,sat}}\}$, and step size $\Delta$.}
\STATE{Find the optimal phase shifts for the IT phase from Proposition \ref{LemmaOne}. }
\WHILE{$t_0 < 1$}
\STATE Obtain $\beta_e^* (t_0)$ from \eqref{Opt_Beta}. 
      
     \STATE{Solve \textbf{P4.3} to obtain $\bar{\bm{t}}^* (t_0)$ and $\bar{\bm{V}}_e (t_0)$}.
     \STATE{Perform Gaussian randomization and obtain $R^* (t_0)$.} 
 \STATE{$t_0=t_0+\Delta$.}
 \ENDWHILE.
 \STATE {Set $t_0^*=\arg \max_{t_0 } R^*$, $\beta_e^*=\beta_e^*(t_0^*)$, $\bm{v}_e^*=\bm{v}_e^* (t_0^*)$, $\bm{e}_u^* =\bm{e}_u^* (t_0^*)$, $\bar{\bm{t}}^*=\bar{\bm{t}}^* (t_0^*) $. }
 \STATE {Set $P_{u,i}^*=e_{u,i}^*/t_i^*,~\forall i$ and extract $\theta_{e,k}^*,~\forall k$ from $\bm{v}_e^*$}.
 \end{algorithmic}
\end{algorithm}

\subsection{Random phase shifts with optimized resource allocation for the PS scheme}
Similar to Section \ref{RndPhase}, we consider the random design of phase shifts for the PS scheme and optimize the resource allocation in the network. With randomly generated phase shifts and after setting $e_{u,i}=P_{u,i}t_i,~\forall i$, we have the following resource allocation problem: 

\begin{equation}\tag{$\textbf{P5}$} 
\begin{aligned}
 \max_{\bm{t},  \bm{e}_u, \beta_e} ~ &\sum_{i=1}^N t_i \log_2(1 +  \frac{{\gamma}_{d,i} } {\sigma_h^2}\frac{e_{u,i}}{t_{i}}), \\ 
\text{s.t.} ~~& \text{C4},~\text{C6},~\text{C11},~\text{C16},~\text{C18}, \\ 
&\text{C20:}~~e_{u,i}+ P_{c,i} t_i \le P_{ps,u,i}  t_0, ~\forall i,
\end{aligned}
\end{equation} 
where 
$P_{ps,u,i} = \min \{\eta P_{h}  | \bm{h}_{u,i}^H \sqrt{\rho} \beta_{e} \bar{\bm{\Theta}}_{e} \bm{h}_r + h_{h,i}  |^2, P_{u,i,sat}\} $.
It can be observed that Proposition \ref{Relationship} also holds for \textbf{P5}. Due to the non-convexity of C20, it is still challenging to solve \textbf{P5}. Hence, we first fix $t_0$ and optimize the time and energy allocation in the IT phase. We then find the optimal value of $t_0$ by searching over $\big( \max \{K \mu / (\eta P_h ||\bm{h}_r ||^2), K \mu/P_{irs,sat}\}, 1 \big)$.

We know from previous discussions that C20 must be met with equality at the optimal solution, i.e., 
\begin{align}
\label{EqualEnergy2}
e_{u,i}+ P_{c,i} t_i = P_{ps,u,i}  t_0.
\end{align}
Consequently, given $t_0$ and $\beta_e$, \textbf{P6} is rewritten as 
\begin{equation}\tag{$\textbf{P6.1}$} 
\begin{aligned}
 \max_{\bar{\bm{t}}} ~ &\sum_{i=1}^N t_i \log_2(1 + d_i \frac{t_0}{t_i} -c_i ), \\ 
\text{s.t.} ~~& \text{C4},~\text{C6}, 
\end{aligned}
\end{equation} 
where $d_i =  P_{ps,u,i} \gamma_{d,i}  / \sigma_h^2$.
The Lagrangian of the above convex problem is given by 
$\mathcal{L}(\bm{\bar{t}}, \zeta)= \sum_{i=1}^N t_i \log_2(1 +  d_i\frac{t_0}{t_i} -c_i )-\zeta \big(t_0+\sum_{i=1}^{N}t_i -1)$,  
where $\zeta$ is  the Lagrange multiplier. 

\begin{proposition}
\label{ProFive}
With fixed $t_0$ and $\beta_e$, the optimal time allocation in the IT phase for the PS scheme is given by 
\begin{align}
t_i^*&=\frac{d_i t_0}{w_i^*+c_i},~\forall i,
\end{align}
where $w_i^* >0$ is the unique solution of 
$\log_2(1+ w_i)-\frac{w_i+c_i}{\ln(2) (1+w_i)}=\zeta^*$,
and $\zeta^*$ is the optimal dual variable.
\end{proposition}

The proof of Proposition \ref{ProFive} is similar to that of Proposition \ref{ProOptiTimeAllocation} and is thus omitted for brevity. Updating $t_0$ by the one-dimensional search method, we can obtain its optimal solution. After that,
the optimal energy allocation can be easily found via \eqref{EqualEnergy2} and Proposition \ref{ProFive}.

\begin{figure}[t]
\centering
\includegraphics[width=3.3 in] {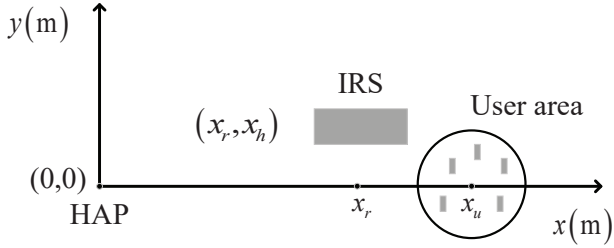}
\caption{Simulation setup for the IRS-assisted WPCN.}
\label{LocationIllustration}
\end{figure}

\section{Performance Evaluation}
\label{Simulation}
In this section, we present numerical results to evaluate the performance of the proposed solutions for the IRS-assisted WPCN. The simulated network topology is a 2-D coordinate system as shown in Fig. \ref{LocationIllustration}, where the coordinates of the HAP and the IRS are given as (0,0) and $(x_r,x_h)$, the users are randomly deployed within a circular area centered at $(x_u, 0)$ with radius 1 m. We follow the channel model considered in \cite{WuSWIPT}. In particular, the large-scale path-loss is modeled as $A (d/{d_0})^{-\alpha}$, where $A$ is the path-loss at the reference distance $d_0=1$ m and set at $A=-10$ dB \cite{Chu}, $d$ denotes the distance between two nodes, and $\alpha$ is the path-loss exponent.  The path-loss exponents of the links between the HAP and users are assumed to be $3.6$ since the users are randomly deployed, while the path-loss exponents of the links between the HAP and IRS and between the IRS and users are set at $2.2$ because the IRS can be carefully deployed to avoid the severe signal blockage. The small-scale fading coefficients are modeled to be Rician fading. In particular, the small-scale channel from the HAP to the IRS can be expressed as 
$\bar{\bm{h}}_r = \sqrt{\frac{\beta_{hap,irs}}{\beta_{hap,irs} +1}} \bar{\bm{h}}_r^{\text{LoS}}  + \sqrt{\frac{1}{\beta_{hap,irs} +1}} \bar{\bm{h}}_r^{\text{NLoS}}$,
where $\beta_{hap,irs}$ is the Rician factor for the HAP-IRS link, $\bar{\bm{h}}_r^{\text{LoS}}$ is the deterministic line of sight (LoS) component, and $\bar{\bm{h}}_r^{\text{NLoS}}$ is the Rayleigh fading component  with circularly symmetric complex Gaussian random variables with zero mean and unit variance. The small-scale channels for the other links are similarly defined. 
 The Rician factors for the HAP-IRS link, the HAP-$U_i$ link, and the IRS-$U_i$ link are set at $\beta_{hap,irs} = 3$, $\beta_{hap,U_i} =0$, and $\beta_{irs,U_i} =3$, respectively.
Unless otherwise stated, other parameters are given as follows: $\rho = 0.8$, $\eta = 0.8$, $\sigma_h^2 = -110$ dBm, $P_{u,sat} = 5$ mW,  $P_{irs,sat} = 0.8$ W,  $\mu= 10$ mW \cite{Huang2019IRS}, $P_{c,i} = 20$ mW, $N=10$, $K=20$,  $P=40$ dBm,  $x_r = 3$ m, $x_h = 0.5$ m, and $ x_u = 6$ m. The scheme with random design of phase shifts, the scheme with random EH time, and the scheme without IRS are used as benchmarks for performance comparisons.

Fig. \ref{HAPPower} shows the influence of the HAP's transmit power on the average system sum-rate. As expected, the average sum-rate is improved with the increase of the HAP's transmit power because the users can harvest more energy when the HAP's transmit power is higher. Further, according to Proposition \ref{ProOptiTime}, the time needed for the IRS's EH in the TS scheme can be reduced when the transmit power of the HAP is increased. This gives more time for the IRS to assist in downlink ET from the HAP to the users, which boosts the harvested energy at the user and consequently improves the sum-rate. As for the PS scheme, increasing the HAP's transmit power results in higher amplitude reflection coefficient according to Proposition \ref{Relationship}, which enhances the users' harvested energy. It can be seen that our proposed schemes with optimized phase shift design outperform the benchmark ones for both the TS and PS schemes. The figure also shows that when $P \le 30$ dBm, there is no gain in using the IRS for improving the performance of WPCN for the PS scheme, which is consistent with what has been noted in Lemma \ref{LemmaFeasible}. When $P \ge 40$ dBm, the average sum-rate achieved by the  PS scheme becomes stable. It is because the maximum harvested power by the IRS and users is limited by the values of their saturation power.

\begin{figure}[t]
\centering
\includegraphics[width=3.5 in] {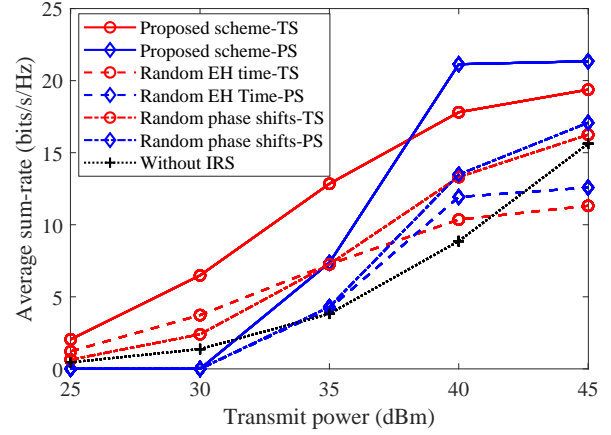}
\caption{Sum-rate versus the HAP's transmit power. }
\label{HAPPower}
\end{figure} 

\begin{figure}[t]
\centering
\includegraphics[width=3.5 in] {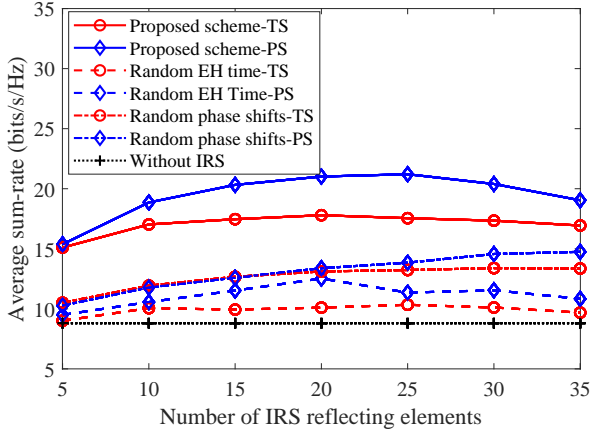}
\caption{Sum-rate versus the number of IRS reflecting elements. }
\label{IRSEelements}
\end{figure}

In Fig. \ref{IRSEelements}, we study the impact of the number of IRS reflecting elements on the average sum-rate. It can be clearly observed that our proposed schemes can achieve a significant gain in terms of the average sum-rate compared with other schemes. As the number of IRS elements increases, the sum-rate achieved by our proposed schemes first increases and then reduces. It is because increasing the number of elements can provide additional transmission links for the ET and IT but also increases the circuit power consumption of the IRS, which thus reduces the users' IT time . If the improved channel power gains can compensate for the reduction of IT time, the  sum-rate can be improved; otherwise, the  sum-rate reduces. This observation indicates that setting an appropriate number of IRS elements is important for performance enhancement.
 For the scheme without IRS, the sum-rate is smallest. It is because  the received power at each user from the HAP through the direct link only is limited, thus more time is required to harvest energy to power its circuit, and the remaining energy and time for the IT is relatively small. It is also worth mentioning that even the  schemes with random phase shifts and the schemes with random EH time can bring performance gains to the  WPCN. That is because the RF energy can still be transferred from the HAP to the users through the reflecting links \cite{Arun}. It endorses the effectiveness of using the IRS for performance enhancement.

\begin{figure}[t]
\centering
\includegraphics[width=3.5 in] {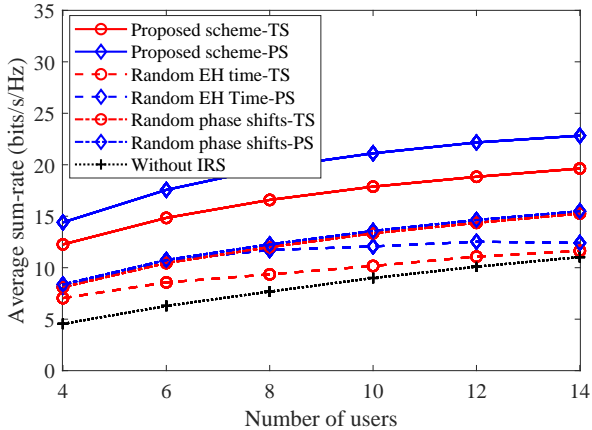}
\caption{Sum-rate versus the number of users. }
\label{NumUsers}
\end{figure} 

\begin{figure}[t]
\centering
\includegraphics[width=3.5 in] {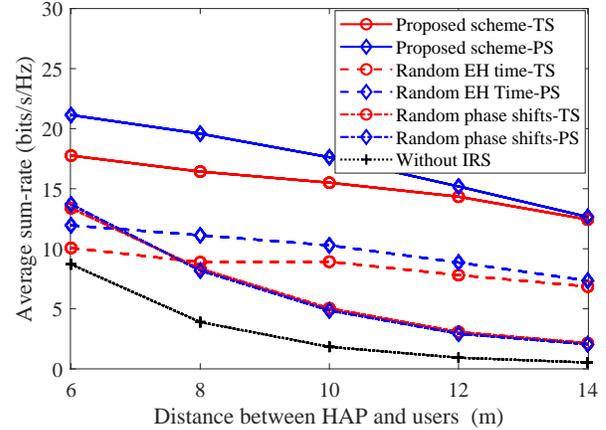}
\caption{Sum-rate versus the distance between HAP and users. }
\label{DistanceUser}
\end{figure}

In Fig. \ref{NumUsers}, we study the effect of the number of network users on the average sum-rate. Again, the proposed IRS-assisted WPCN with optimal phase shift design notably outperforms the other schemes. It can be observed that the average sum-rate is increasing with the number of users because more energy can be harvested with the increase of  the number of users.
 Nevertheless, the average sum-rate does not increase when the number of users reaches a high number, e.g., over 10 users. The reason for this observation is that adding new users implies that more time is needed for the IT phase, which in consequence decreases the ET phase duration. Shorter ET duration in the TS scheme means that less time will be left for the IRS to assist in the downlink ET. In the PS scheme, the IRS needs to decrease its  amplitude reflection coefficient $\beta_e$ to compensate for the loss of energy incurred by shortening the ET duration. Therefore, the gain brought by incrementing the number of users is neutralized by shortened ET time and the average sum-rate converges to an upper bound.

\begin{figure}[t]
\centering
\includegraphics[width=3.5 in] {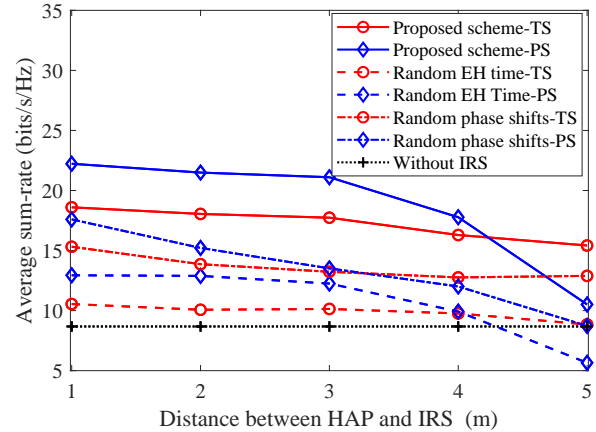}
\caption{Sum-rate versus the distance between HAP and IRS. }
\label{DistanceIRS}
\end{figure}

Next, we investigate the effect of users' locations on the sum-rate performance. As shown in Fig. \ref{DistanceUser}, increasing $x_u$ results in sum-rate reduction because as $x_u$ increases, the users move further from both the HAP and IRS. Therefore, the signals received by the users in the ET phase from both the HAP and the IRS become weaker. Similarly, the signals received by the HAP in the uplink IT also get weaker. Once again, the proposed schemes can significantly outperform the other benchmark schemes.

Finally, we investigate the impact of the IRS's location on the sum-rate performance in Fig. \ref{DistanceIRS}. It can be observed that increasing the distance between the HAP and IRS reduces the  sum-rate for the proposed schemes. It is because as the distance increases, the IRS has to spend more time to harvest sufficient energy to power its operations, which thus reduces the IT time for users. Compared to the TS scheme,  the PS scheme is more susceptible to the IRS's location. The reason is that for the PS scheme, as the distance increases, not only the users' IT time but also the amplitude reflection coefficient at the IRS will be reduced.

\section{Conclusions}
\label{Conclusion}
This paper has proposed the hybrid-relaying scheme empowered by a self-sustainable IRS to enhance the performance of WPCN, where the IRS is deployed to improve the efficiency of downlink ET from the HAP to a number of users and uplink IT from the users to the HAP. In addition, we have proposed the TS and PS schemes for the IRS to harvest sufficient energy from the HAP to power its operations and investigated system sum-rate maximization problems for both schemes. To address the non-convexity of each formulated problem, we have developed the two-step algorithms to efficiently obtain the near-optimal solution with satisfying accuracy. The special problems with random phase shifts have  also been investigated to revel the structure of time and energy allocation.
Then, we have performed simulations to evaluate the superiority of our proposed schemes, which have shown that our proposed schemes can achieve remarkable sum-rate gain compared to the baseline WPCN without IRS. From simulation results, we have also observed that the PS scheme can achieve a  better performance than the TS scheme if the transmit power at the HAP is large enough or the channel between the HAP and IRS is strong. However, compared to the PS scheme, the TS scheme can be more widely applied  because it is free from the constraint defined in Lemma \ref{LemmaFeasible} for the PS scheme.

\appendices

\section{Proof of Proposition \ref{LemmaOne}}
\label{App:LemmaOne}
 It is straightforward that $R_i$ is an increasing function with respect to $|\bm{g}_r^H \bm{\Theta}_{d,i} \bm{g}_{u,i} + g_{h,i}|^2$ for $i=1,\ldots,N$. Therefore, the optimal solution of \textbf{P1} is found when $|\bm{g}_r^H \bm{\Theta}_{d,i} \bm{g}_{u,i} + g_{h,i}|^2,\forall i$ is maximized. In addition, $|\bm{g}_r^H \bm{\Theta}_{d,i} \bm{g}_{u,i} + g_{h,i}|^2$ only depends on $\bm{\Theta}_{d,i}$. As a result, for any given and feasible $\bm{t}$ and $\bm{P}_u$, maximizing the objective function of \textbf{P1} with respect $\bm{\Theta}_{d,i}$  is equivalent to solving the following problem for $i=1,\ldots,N$
 \begin{equation}\tag{$\textbf{P-A}$} 
\begin{aligned}
  \max_{\bm{\Theta}_{d,i}}~~ &|\bm{g}_r^H \bm{\Theta}_{d,i} \bm{g}_{u,i} + g_{h,i}|^2, \\ 
\text{s.t.}~~  &|v_{d,i,k}|=1, ~\forall k.
\end{aligned}
\end{equation}

The objective function $|\bm{g}_r^H \bm{\Theta}_{d,i} \bm{g}_{u,i} + g_{h,i}|^2 $ can be rewritten as $|\bm{g}_r^H \bm{\Theta}_{d,i} \bm{g}_{u,i}|^2 + |g_{h,i}|^2 + 2 |\bm{g}_r^H \bm{\Theta}_{d,i} \bm{g}_{u,i}| |g_{h,i}| \cos{\alpha}$, where 
$\alpha = \arctan {\frac{ \text{Im}(\bm{g}_r^H \bm{\Theta}_{d,i} \bm{g}_{u,i})} {\text{Re} (\bm{g}_r^H \bm{\Theta}_{d,i} \bm{g}_{u,i})}  - \arctan { \frac{\text{Im} (g_{h,i})} {\text{Re}(g_{h,i})}  } }$.
It is obvious that the maximum of $|\bm{g}_r^H \bm{\Theta}_{d,i} \bm{g}_{u,i} + g_{h,i}|^2 $ is achieved if $\alpha = 0$, i.e., $\arg(\bm{g}_r^H \bm{\Theta}_{d,i} \bm{g}_{u,i})$ = $\arg(g_{h,i}) \overset{{\Delta}}{=} \omega$. Denoting $\bm{v}_{d,i} = [v_{d,i,1},\ldots,v_{d,i,K}]^H$ and $\bm{\phi}_i =  \text{diag}(\bm{g}_r^H) \bm{g}_{u,i}$, we have $\bm{g}_r^H \bm{\Theta}_{d,i} \bm{g}_{u,i} = \sqrt{\rho} \bm{v}_{d,i}^H \bm{\phi}_i$. Then, \textbf{P-A} can be rewritten as 
 \begin{equation}\tag{$\textbf{P-B}$} 
\begin{aligned}
& \max_{\bm{v}_{d,i} } |\bm{v}_{d,i}^H \bm{\phi}_i|^2, \\ 
\text{s.t.} \ & |v_{d,i,k}|=1, ~\forall k, \\
& \arg(\bm{v}_{d,i}^H \bm{\phi}_i) = \omega.
\end{aligned}
\end{equation}

According to \cite{WuIRS}, the optimal solution to \textbf{P-B} can be expressed as
$\bm{v}_{d,i}^* = e^{j(\omega -\arg(\bm{\phi}_i))} = e^{j (\omega - \arg(\text{diag}(\bm{g}_r^H) \bm{g}_{u,i}  ) )  }$.
Then, the optimal phase shift for the $k$-th element of the IRS is expressed as 
$\theta_{d,i,k}^* = \omega - \arg(\bm{g}_{r,k}^H) - \arg(\bm{g}_{u,i,k}).$
This completes the proof of Proposition \ref{LemmaOne}.

\section{Proof of Proposition \ref{ProOptiTime}}
\label{App:ProOptiTime}
It can be verified that the objective function of \textbf{P2.1} is an increasing function with respect to $t_i$ and $e_{u,i}$ for $i=1,\ldots,N$. Therefore, at the optimal solution, C10 must be met with equality. The constraint C1 must also be satisfied with equality, because otherwise we can decrease $\tau_0$ and increase $\tau_1$, which results in more harvested energy at the users and larger transmit energy $e_{u,i},\forall i$. We can also observe that the right hand side of C1 is increasing with respect to $\tau_0$. Thus, the constraint C3 must be met with equality at the optimal solution because otherwise we can always increase $\tau_0$ as a result of which $\tau_1$ and users' harvested energy can also be increased. Similarly, the constraint C4 must also be an equality at the optimal solution as otherwise we can increase $t_0$, leading to the increase of $\tau_0$ and $\tau_1$.
Based on the three equalities from the constraints C1, C3 and C4, we can straightforwardly  obtain the optimal value of $\tau_0$ as given by \eqref{Optitau0}.

\section{Proof of Proposition \ref{ProOptiTimeAllocation}}
\label{App:ProOptiTimeAllocation}
The dual function of \textbf{P3.1} is given by 
$\mathcal{G}(\xi) = \max_{\bar{\bm t} \ge 0, \tau_1 \ge 0} \mathcal{L} (\bar{\bm t}, \tau_1, \xi)$.
Karush-Kuhn-Tucker (KKT) conditions are both necessary and sufficient for the optimality of \textbf{P3.1} \cite{BoydOne}, which are given by 
\begin{align}
\label{Partialti}
&\frac{\partial{\mathcal{L}}} {\partial{t_i}}  = \log_2 \Big(1+ \frac{a_i +b_i \tau_1^*} {t_i^* } - c_i \Big) \nonumber \\
&~~- \frac{\frac{a_i +b_i \tau_1^* }{t_i^*  } }{ \ln(2) \Big(1+ \frac{a_i + b_i \tau_1^*}{t_i^* } -c_i \Big) } - \xi^* = 0, \\
\label{Partitau1}
&\frac{\partial{\mathcal{L}}} {\partial{\tau_1}}=\sum_{i=1}^N \frac{b_i}{\ln(2) \Big(1+ \frac{a_i +b_i \tau_1^*} {t_i^* } - c_i \Big) } - \xi^* = 0, \\
\label{Slack}
&\xi^* \Big[\tau_0^* + \tau_1^* + \sum_{i=1}^N t_i^* -1  \Big] = 0.
\end{align}

Setting $z_i = \frac{a_i + b_i \tau_1}{t_i } - c_i$ and substituting it into \eqref{Partialti} and \eqref{Partitau1}, we have  
\begin{align}
\label{PartialtiNew}
&\log_2(1+z_i) - \frac{z_i + c_i}{ \ln(2) (1+z_i) } = \xi^*,\\
\label{Partialtau1New}
& \sum_{i=1}^N \frac{b_i}{\ln(2) (1+z_i) } = \xi^*.
\end{align}
It is straightforward to verify that the left hand side of \eqref{PartialtiNew} is a strictly increasing function with respect to $z_i >0$. Hence, there exists a unique solution, denoted by $z_i^*$, satisfying \eqref{PartialtiNew}. From \eqref{Partialtau1New}, we can observe that $\xi^*$ is upper-bounded by $\frac{1}{\ln(2)}\sum_{i=1}^N b_i$ and can be thus found by the bisection method. Also, \eqref{Partialtau1New} indicates that $\xi^*>0$. Having $\tau_0^* + \tau_1^* + \sum_{i=1}^N t_i^* = 1$ from \eqref{Slack} and $z_i^* = \frac{a_i +b_i\tau_1^*}{t_i^* } - c_i$, \eqref{Opttau1} and \eqref{Optti} are obtained with some simple mathematical calculations. This thus proves Proposition \ref{ProOptiTimeAllocation}.

\section{Proof of Lemma \ref{LemmaFeasible}}
\label{App:LemmaFeasible}
First of all, for the IRS to be able to assist in  downlink ET and uplink IT, we must have 
\begin{align}
\label{condition1}
    K\mu  \le  P_{irs,sat} t_0^*
\end{align}
according to C16. Otherwise if $K\mu  >  P_{irs,sat} t_0^*$, the IRS cannot harvest enough energy to power its circuit operations even if the harvested power reaches its maximum value (i.e., saturation power). Furthermore, at optimality, the received power at the energy harvester of the IRS must not be greater than the saturation power, because otherwise, the amount of the reflected power by the IRS can be increased by increasing the amplitude reflection coefficient, without affecting the amount of harvested power at the IRS. Therefore, we must have 
    $\eta P_h (1-\beta_e^{*2}) ||\bm{h}_r ||^2 \le P_{irs,sat}$.
Therefore, $\min \{\eta P_h (1-\beta_e^{*2}) ||\bm{h}_r||^2, P_{irs,sat}\}=\eta P_h (1-\beta_e^{*2}) ||\bm{h}_r||^2$. Now, according to the energy causality constraint of the IRS in C16, we have
    $K\mu  \le  \eta P_h (1-\beta_e^{*2}) ||\bm{h}_r||^2  t_0^*$.
Thus, $\beta_e^*$ is upper-bounded by 
\begin{align}
    \label{upperbound}
    \beta_e^* \le \sqrt{1-{K\mu }/( \eta P_h  ||\bm{h}_r||^2 t_0^*)}.
\end{align} 
To ensure a feasible value for $\beta_e^*$, the following condition must be satisfied:
\begin{align}
\label{condition2}
    K\mu  \le  \eta P_h ||\bm{h}_r||^2 t_0^*. 
\end{align} 
From \eqref{condition1} and \eqref{condition2} and the fact that ${t_0^*} < 1$, we obtain
$K \mu < \min\{\eta P_h ||\bm{h}_r||^2, P_{irs,sat}\}$.
Lemma \ref{LemmaFeasible} is thus proved. 
At the optimal solution, the amplitude reflection coefficient must be set to its upper-bound to maximize the amount of reflected power from the IRS. Therefore, according to \eqref{upperbound}, $\beta_e^*$ is calculated as 
$\beta_e^* = \sqrt{1-{K\mu }/( \eta P_h  ||\bm{h}_r||^2 t_0^*)}$,
where $\max \{\frac{K \mu }{\eta P_h  ||\bm{h}_r ||^2 }, \frac{K\mu }{P_{irs,sat}} \} < t_0^* < 1$ according to \eqref{condition1} and \eqref{condition2}. This thus proves Proposition \ref{Relationship}.


\begin{thebibliography}{1}

\bibitem{IoT}
Cisco edge-to-enterprise IoT analytics for electric utilities. Available
Online: \url{https://www.cisco.com/c/en/us/solutions/collateral \
 /data-center-virtualization/big-data/solution-overview-c22-740248.html}, Feb. 2018.

\bibitem{RFSurvey}
X. Lu, P. Wang, D. Niyato, D. I. Kim, and Z. Han, ``Wireless networks with RF energy harvesting: A contemporary survey," \emph{ IEEE Commun. Surv. Tut.}, vol. 17, no. 2, pp. 757-789, Secondquarter 2015.

\bibitem{JuOne}
H. Ju and R. Zhang, ``Throughput maximization in wireless powered communication networks,'' \emph{IEEE Trans. Wireless Commun.}, vol. 13, no. 1, pp. 418-428, Jan. 2014.

\bibitem{Abbas}
P. Ramezani and A. Jamalipour, ``Toward the evolution of wireless powered communication networks for the future Internet of Things,'' \emph{IEEE Network}, vol. 31, no. 6, pp. 62-69, Nov./Dec. 2017.

\bibitem{HeChen}
H. Chen, Y. Li, J. L. Rebelatto, B. F. Uchoa-Filho, and B. Vucetic, ``Harvest-then-cooperate: Wireless-powered cooperative communications,'' \emph{IEEE Trans. Signal Process.}, vol. 63, no. 7, pp. 1700-1711, Apr., 2015.

\bibitem{JuThree}
H. Ju and R. Zhang, ``Uer cooperation in wireless powered communication networks,'' in \emph{Proc. IEEE GLOBECOM}, Austin, TX, USA, Dec. 2014, pp. 1430-1435. 

\bibitem{ZengTwo}
 Y. Zeng, H. Chen, and R. Zhang, ``Bidirectional wireless information and power transfer with a helping relay,`` \emph{IEEE Commun. Letters}, vol. 20, no. 5, pp. 862-865, May 2016.

 \bibitem{CJZhong}
C. Zhong, H. A. Suraweera, G. Zheng, I. Krikidis, and Z. Zhang, ``Wireless information and power transfer with full duplex relaying,`` \emph{IEEE Trans. Commun.}, vol. 62, no. 10, pp. 3447-3461, Oct. 2014.


\bibitem{LyuOne}
B. Lyu, D. T. Hoang, and Z. Yang,  ``User Cooperation in wireless-powered backscatter communication networks,'' \emph{IEEE Wireless Commun. Lett.}, vol. 8, no. 2, pp. 632-635, Apr. 2019.

\bibitem{Kim2017Hybrid}
S. H. Kim and D. I. Kim, ``Hybrid backscatter communication for wireless-powered heterogeneous networks,'' \emph{IEEE Trans. Wireless Commmun.}, vol. 16, no. 10, pp. 6557-6570, Oct. 2017.

\bibitem{Gong}
S. Gong, X. Huang, J. Xu, W. Liu, P. Wang, and D. Niyato, ``Backscatter relay communications powered by wireless energy beamforming,'' \emph{IEEE Trans. Commun.}, vol. 66, no. 7, pp. 3187-3200, July 2018.

\bibitem{LiuOne}
V. Liu, A. Parks, V. Talla, S. Gollakota, D. Wetherall, and J. R. Smith, ``Ambient backscatter: Wireless communication out of thin air,'' in \emph{Proc. SIGCOMM}, pp. 39-50, Hong Kong, Aug. 2013.

\bibitem{Wu2020Survey}
Q. Wu and R. Zhang, ``Towards smart and reconfigurable environment: Intelligent reflecting surface aided wireless network,''  \emph{IEEE Commun. Mag.}, vol. 58, no. 1. pp 106-112, Jan. 2020.

\bibitem{Gong2019Towards}
S.~Gong \emph{et al.} ``Towards smart radio environment for wireless communications via intelligent reflecting surfaces: A comprehensive survey,'' \emph{IEEE Commun. Surveys Tuts.,}, doi: 10.1109/COMST.2020.3004197, 2020.

\bibitem{Renzo}
M. D. Renzo \emph{et al.}, ``Smart radio environments empowered by reconfigurable AI meta-surfaces: An idea whosetime has come,''  \emph{EURASIP J. Wireless Commun. Netw.}, vol. 129, pp. 1-20, May 2019.

\bibitem{Huang2018Conf}
C. Huang,  G. C. Alexandropoulos, A. Zappone, M. Debbah, and C. Yuen, ``Energy efficient multi-user MISO communication using low resolution large intelligent surfaces,'' in \emph{Proc. IEEE GLOBECOM Workshops}, Abu Dhabi, United Arab Emirates, 2018, pp. 1-6.

\bibitem{Huang2019IRS}
C. Huang, A. Zappone, G. C. Alexandropoulos, M. Debbah, and C. Yuen,  ``Reconfigurable intelligent surfaces for energy efficiency in wireless communication,'' \emph{IEEE Trans. Wireless Commun.}, vol. 18, no. 8, pp. 4157-4170, Jun. 2019.

\bibitem{Taha}
A.~Taha, M. Alrabeiah, and A. Alkhateeb, ``Enabling large intelligent surfaces with compressive sensing and deep learning,''
 Available Online: \url{https://arxiv.org/pdf/1904.10136.pdf}, Apr. 2019.


 \bibitem{HuangDeep}
C. Huang, R. Mo and C. Yuen, ``Reconfigurable intelligent surface assisted multiuser MISO systems exploiting deep reinforcement learning,'' \emph{IEEE J. Sel. Areas
Commun.}, vol. 38, no. 8, pp. 1839-1850, Aug. 2020.

\bibitem{Mishra}
D. Mishra and H. Johansson, ``Channel estimation and low-complexity beamforming design for passive intelligent surface assisted MISO wireless energy transfer,'' in \emph{Proc. IEEE ICASSP}, Brighton, UK, May 2019, pp. 4659–4663.

\bibitem{Yu}
X. Yu, D. Xu, and R. Schober,  ``MISO wireless communication systems via intelligent reflecting surfaces,'' in \emph{Proc. IEEE/CIC ICCC}, Changchun, China, Aug. 2019, pp. 735–740. 

\bibitem{WuIRS}
Q. Wu and R. Zhang, ``Intelligent reflecting surface enhanced wireless network via joint active and passive beamforming,'' \emph{IEEE Trans. Wireless Commun.}, vol. 18, no. 11, pp. 5394-5409, Nov. 2019.

\bibitem{WuSWIPT}
Q. Wu and R. Zhang, ``Weighted sum power maximization for intelligent reflecting surface aided SWIPT,''  \emph{IEEE Wireless Commun. Lett.}, vol. 9, no. 5,  pp. 586-590, May 2020.

\bibitem{PanSWIPT}
C. Pan,   H. Ren, K. Wang, M. Elkashlan, A. Nallanathan, J. Wang, and L. Hanzo ``Intelligent reflecting surface aided MIMO broadcasting for simultaneous wireless information and power transfer,''  \emph{IEEE J. Sel. Area. Commun.}, vol. 38, no. 8, pp. 1719-1734, Aug. 2020.


\bibitem{Chu}
Z. Chu, W. Hao, P. Xiao, and J. Shi, ``Intelligent reflect surface aided multi-antenna
secure transmission,'' \emph{IEEE Wireless Commun. Lett.}, vol. 9, no. 1, pp. 108-112, Jan. 2020.



\bibitem{Guo}
C. Guo, Y. Cui, F. Yang, and L. Ding, ``Outage probability analysis and minimization in intelligent reflecting surface-assisted MISO systems,'' \emph{IEEE Commun. Lett.}, vol. 24, no. 7, 1563-1567, July 2020.


\bibitem{GongTwo}
Y. Zou, Y. Liu, S. Gong, W. Cheng, D. T. Hoang, and D. Niyato, ``Joint energy beamforming and optimization for intelligent reflecting surface enhanced communications,'' in \emph{Proc. WCNC Workshops}, Seoul, South Korea, May 2020, pp. 1-6. 

\bibitem{Zhao}
J.~Zhao and Y. Liu, ``A survey of intelligent reflecting surfaces (IRSs): Towards 6G wireless communication networks,''
 Available Online: \url{https://arxiv.org/pdf/1907.04789.pdf}, Nov. 2019.


\bibitem{LyuTwo}
B. Lyu, D. T. Hoang, S. Gong, and Z. Yang, ``Intelligent reflecting surface assisted wireless powered communication networks,'' in \emph{Proc. WCNC Workshops}, Seoul, South Korea, May 2020, pp. 1-6. 

\bibitem{SuzhiBi}
Y. Zheng, S. Bi, Y. J. Zhang, Z. Quan, and H. Wang, ``Intelligent reflecting surface enhanced user cooperation in wireless powered communication networks,'' \emph{IEEE Wireless Commun. Lett.},  vol. 9, no. 6, pp. 901-905, June 2020.

\bibitem{HuangHolo}
C.~Huang \emph{et al.}, ``Holographic MIMO surfaces for 6G wireless networks: Opportunities, challenges, and trends,'' \emph{IEEE Wireless Commun. Mag.}, doi:10.1109/MWC.001.1900534, 2020.
 
 \bibitem{Nasir}
A. A. Nasir, X. Zhou, S. Durrani, and R. A. Kennedy, ``Relaying protocols for wireless energy harvesting and information processing," \emph{ IEEE Trans. Wireless Commun.}, vol. 12, no. 7, pp. 3622-3636, Jul. 2013.

 \bibitem{Panos}
P. N. Alevizos and A. Bletsas, ``Sensitive and Nonlinear Far-Field RF Energy Harvesting in Wireless Communications," \emph{ IEEE Trans. Wireless Commun.}, vol. 17, no. 6, pp. 3670-3685, Jun. 2018.

\bibitem{Dong}
Y. Dong, M. J. Hossain, and J. Cheng, ``Performance of wireless powered amplify and forward relaying over Nakagami-m fading channels with nonlinear energy harvester," \emph{ IEEE Commun. Lett.}, vol. 20, no. 4, pp. 672-675, Apr. 2016.

\bibitem{Niyato}
A. El Shafie, D. Niyato, and N. Al-Dhahir, ``Security of an ordered-based distributive jamming scheme,'' \emph{IEEE Commun. Lett.} vol. 21 no. 1 pp. 72-75 Jan. 2017. 

\bibitem{Schober}
S. Pejoski, Z. Hadzi-Velkov, and R. Schober, ``Optimal power and time allocation for WPCNs with piece-wise linear EH model,'' \emph{IEEE Wireless Commun. Lett.}, vol. 7, no. 3, pp. 364-367, June 2018.

\bibitem{ShuguangCui}
Z. Wang, L. Liu, and S. Cui, ``Channel estimation for intelligent
reflecting surface assisted multiuser communications,'' \emph{Proc. WCNC}, Seoul, South Korea, May 2020, pp. 1-6. 

\bibitem{JunZhao}
Y. Gao \emph{et al.}, ``Reconfigurable intelligent surface for MISO systems with proportional rate constraints,'' \emph{Proc. ICC}, Dublin, Ireland, June 2020, pp. 1-7.


\bibitem{Boshkovska}
E. Boshkovska \emph{et al.}, ``Practical non-linear energy harvesting model and resource allocation for SWIPT systems,'' \emph{IEEE Commun. Lett.}, vol. 19, no. 12, pp. 2082-2085, Dec. 2015.

\bibitem{Alexandropoulos}
D. Mishra and G. C. Alexandropoulos, ``Transmit precoding and receive power splitting for harvested power maximization in MIMO SWIPT systems," \emph{IEEE Trans. Green Commun. Netw.}, vol. 2, no. 3, pp. 774-786, Sept. 2018.

\bibitem{GuangyueLu}
C. Lu, L. Shi, and Y. Ye, ``Maximum throughput of TS/PS scheme in an AF relaying network with non-linear energy harvester," \emph{IEEE Access}, vol. 6, pp. 26617-26625, 2018.



\bibitem{Yang}
H. Yang \emph{et al.}, ``Design of resistor-loaded reflectarray elements for both amplitude and phase control," \emph{ IEEE  Antennas Wireless Propag. Lett.}, vol. 16, pp. 1159–1162, 2017.

\bibitem{Derrick}
S. Hu, Z. Wei, Y. Cai, D. W. K. Ng, and J. Yuan, ``Sum-rate maximization for multiuser MISO
downlink systems with self-sustainable IRS,'' Available Online: \url{https://arxiv.org/pdf/2005.11663.pdf}, May 2020.

\bibitem{Luo}
Z. Q. Luo, W.-K. Ma, A. M.-C. So, Y. Ye, and S. Zhang, ``Semidefinite relaxation of quadratic optimization problems,'' \emph{IEEE Signal Process.}, vol. 27, no. 3, pp. 20-34, May 2010.

\bibitem{BoydOne}
S. Boyd and L. Vandenberghe, \emph{Convex Optimization}. Cambridge University Press, 2004.

\bibitem{BoydTwo}
Michael Grant \emph{et al.}, ``CVX: Matlab software for disciplined convex programming,''  Available Online: http://cvxr.com/cvx, September 2013.

\bibitem{SPR}
A. M.-C. So, J. Zhang, and Y. Ye, ``On approximating complex quadratic
optimization problems via semidefinite programming relaxations,''  \emph{Mathematical
Programming,} vol. 110, no. 1, pp. 93–110, Jun. 2007.

\bibitem{Arun}
V. Arun \emph{et al.},  ``RFocus: Practical beamforming for small devices''.   Available Online: \url{https://arxiv.org/pdf/1912.07794.pdf}, May 2019.


\end{thebibliography}
\end{document}